\documentclass[aps,prd,10pt,onecolumn,notitlepage,superscriptaddress,nofootinbib,longbibliography,showkeys]{revtex4-2}
\usepackage{amsmath,amssymb,amsfonts,amsthm,mathtools,bm}
\usepackage{graphicx}
\usepackage[margin=1in]{geometry}
\graphicspath{{figs/}}
\usepackage{microtype}
\usepackage{enumitem}
\usepackage{xcolor}
\usepackage[colorlinks=true,allcolors=blue]{hyperref}
\usepackage{orcidlink}

\newtheorem{theorem}{Theorem}
\newtheorem{lemma}[theorem]{Lemma}
\newtheorem{proposition}[theorem]{Proposition}
\newtheorem{corollary}[theorem]{Corollary}

\newcommand{\E}{\mathbb E}
\newcommand{\Prob}{\mathbb P}
\newcommand{\Tr}{\operatorname{Tr}}
\newcommand{\id}{\operatorname{id}}

\newcommand{\D}{\mathcal D}
\newcommand{\Ochan}{\mathcal O}

\newcommand{\norm}[1]{\lVert#1\rVert}

\newcommand{\bbone}{\mathbf 1}
\newcommand{\Fopt}{F_{\rm e}^{\rm opt}}
\newcommand{\rec}{\mathrm{rec}}
\newcommand{\Prb}{\mathbb P}

\hypersetup{
  pdftitle={Microscopic Side Information Controls Ordered Hayden--Preskill Recovery},
  pdfauthor={Joao V. R. Alencar; Allan R. P. Moreira; Joao B. R. Silva},
  pdfsubject={Quantum information recovery under ordered quantum deletion},
  pdfkeywords={Hayden--Preskill protocol, quantum deletion channels, quantum recovery, decoupling, planted subsequences, microscopic side information}
}

\begin{document}

\title{Microscopic Side Information Controls Ordered Hayden--Preskill Recovery}

\author{Jo\~ao V. R. Alencar\,\orcidlink{0009-0008-3042-5807}}
\email{joaoalencar@alu.ufc.br}
\affiliation{Department of Teleinformatics Engineering, Federal University of Cear\'a (UFC), Fortaleza, CE, 60440-900, Brazil}

\author{Allan R. P. Moreira\,\orcidlink{0000-0002-6535-493X}}
\email{allan.moreira@fisica.ufc.br}
\affiliation{Secretaria da Educa\c{c}\~ao do Cear\'a (SEDUC), Coordenadoria Regional de Desenvolvimento da Educa\c{c}\~ao (CREDE 9), Horizonte, Cear\'a, 62880-384, Brazil}
\affiliation{Postgraduate Program in Electrical and Computer Engineering, Federal University of Cear\'a, Sobral, Cear\'a, 62010-560, Brazil}

\author{Jo\~ao B. R. Silva\,\orcidlink{0000-0002-4004-0926}}
\email{joaobrs@ufc.br}
\affiliation{Department of Teleinformatics Engineering, Federal University of Cear\'a (UFC), Fortaleza, CE, 60440-900, Brazil}


\begin{abstract}
In the Hayden--Preskill protocol, the decoder is usually assumed to know the microscopic identity of the collected output qubits. We study what happens when these labels are unavailable and only the relative order of the received qubits is preserved. The resulting order-preserving deletion channel maps an $n$-qubit scrambled register to a subsequence of length $\ell$. For a diary of fixed size $k$, we prove that the optimal entanglement fidelity converges to the no-output value $4^{-k}$ when $\ell=o(n^{2/3})$, uniformly over the scrambling unitary. For a Haar-random scrambling unitary, it converges to one when $\ell=\omega(n^{2/3})$ and $\ell=o(n)$. Monotonicity then gives the fixed-error recovery scale $\ell_{\mathrm{rec}}=\Theta(n^{2/3})$. We also consider partial position information obtained by dividing the register into $B$ consecutive blocks and revealing the block of origin of each received qubit. The recovery scale becomes $\ell_{\mathrm{rec}}(n,B)\asymp n^{2/3}B^{-1/3}$ for $B\leq\sqrt n$ and $\ell_{\mathrm{rec}}(n,B)\asymp n/B$ for $B\geq\sqrt n$. The exponent $2/3$ is traced to rank-aligned coincidences between random subsequences, which control both the converse and the recovery argument. Thus, even purely classical information about the origin of the received subsystems can change the amount of quantum output required for Hayden--Preskill recovery.
\end{abstract}


\maketitle

\section{Introduction}
\label{sec:introduction}

Quantum scrambling spreads initially localized information across many degrees of freedom and encodes it in nonlocal correlations. In black-hole physics, this mechanism underlies the Hayden--Preskill thought experiment. After the Page time, a small diary deposited in an old black hole can be reconstructed from a relatively small amount of subsequently emitted radiation, provided that the internal dynamics is sufficiently mixing \cite{Page1993,HaydenPreskill2007,SekinoSusskind2008}. The protocol has since become a standard model for quantum information transfer through chaotic dynamics and for the error-correcting structure associated with evaporation.

Later work has examined the complexity of decoding and developed explicit recovery procedures \cite{HarlowHayden2013,YoshidaKitaev2017,YoshidaYao2019,Yoshida2022,NakataMatsuuraKoashi2025}. Variants incorporating noise, finite temperature, conserved quantities, and structured many-body dynamics have extended the original random-unitary setting to more general channels and physical systems \cite{BaoKikuchi2020,LiWang2022,NakataTezuka2024,RamppClaeys2024}. Recovery has also been connected to operator growth and channel entanglement \cite{HosurEtAl2016}, while teleportation-based protocols have provided experimentally accessible probes of Hayden--Preskill scrambling \cite{LandsmanEtAl2019,SekiEtAl2025}. These developments generally assume that the decoder knows which output subsystems were collected.

The identities of those subsystems are part of the classical information supplied to the decoder. A receiver who obtains qubits $i_1,\ldots,i_\ell$ together with their indices has access to a different channel from one who receives the same qubits after the indices have been erased. In the latter case, the channel averages over all microscopic embeddings compatible with the observed output. The two receivers hold systems of the same Hilbert-space dimension, but their optimal recovery fidelities can differ parametrically.

We study this distinction through an ordered deletion channel. An $n$-qubit register is scrambled, a uniformly random set of $\ell$ positions is retained, and the surviving qubits are presented in their original relative order. Their absolute positions are unavailable to the receiver. The model preserves the relative canonical order inherited from the scrambled register while erasing the microscopic locations of the surviving systems.

For a diary of fixed size $k$, the resulting recovery scale differs from that of the usual labeled Hayden--Preskill setting. The optimal entanglement fidelity satisfies
\begin{equation}
\ell=o(n^{2/3})
\quad\Longrightarrow\quad
F_{\rm e}^{\rm opt}\to4^{-k},
\qquad
\ell=\omega(n^{2/3}),\ \ell=o(n)
\quad\Longrightarrow\quad
F_{\rm e}^{\rm opt}\to1.
\label{eq:intro-regimes}
\end{equation}
The first convergence is uniform over the scrambling unitary. Since $F_{\rm e}^{\rm opt}$ is already optimized over all recovery maps, no decoder asymptotically exceeds the baseline $4^{-k}$. The second convergence holds in probability for a Haar-random scrambling unitary. Monotonicity under further deletion then yields the fixed-error law
\[
\ell_{\mathrm{rec}}=\Theta(n^{2/3}).
\]
In the standard labeled-subsystem setting, the corresponding recovery length is $O(k)$ and therefore remains constant when $k$ is fixed. Ordered deletion makes the required output grow as $n^{2/3}$. This scale is mesoscopic. It diverges with the size of the register while remaining parametrically smaller than the full system.

The exponent $2/3$ comes from the order statistics of random subsets. For ranks away from the endpoints, the position of the $a$th survivor fluctuates over a window of width of order $n/\sqrt{\ell}$. Two independent subsequences therefore place their $a$th survivors at the same microscopic site with probability of order $\sqrt{\ell}/n$. Summing over the bulk ranks gives the alignment strength
\begin{equation}
V_{n,\ell}\asymp\frac{\ell^{3/2}}{n}.
\label{eq:intro-alignment}
\end{equation}
The condition $V_{n,\ell}\asymp1$ gives $\ell\asymp n^{2/3}$. The same statistic appears in both parts of the argument. Below the recovery scale, it controls the purity of the channel Choi state. Above the scale, it governs the separation of a planted-subsequence model from the competing product laws.

Partial position information produces a family of intermediate recovery problems. We divide the register into $B$ consecutive blocks and reveal the block containing each surviving qubit, while keeping its position inside the block hidden. The recovery length obeys
\begin{equation}
\ell_{\mathrm{rec}}(n,B)\asymp
\begin{cases}
n^{2/3}B^{-1/3},& B\leq\sqrt n,\\[1mm]
n/B,& B\geq\sqrt n.
\end{cases}
\label{eq:intro-phase}
\end{equation}
When a typical block contains many survivors, recovery relies on the accumulation of weak rank information within each block. When the blocks are smaller, their labels almost localize the surviving systems, and the recovery scale is determined by the microscopic block size $n/B$. The crossover at $B\asymp\sqrt n$ separates these two mechanisms.

The converse starts from an exact identity for the normalized Choi purity of the ordered deletion channel. This purity is the exponential moment of the number of rank-aligned coincidences between two independent subsequences. Factorial-moment estimates show that the normalized purity tends to one in the subcritical regime and provide the mixing estimate used in the converse. For achievability, a computational-basis measurement reduces the problem to distinguishing the planted law from $P_XQ_Y$, uniformly over the output law $Q_Y$, with $P_X$ fixed as the uniform input marginal. A weighted alignment statistic gives the required separation \cite{Sibson1969}. Data processing for sandwiched R\'enyi divergences and joint state--channel decoupling then lift the classical separation to quantum recovery \cite{FrankLieb2013,Beigi2013,LeditzkyRouzeDatta2016,ChengDupuisGao2024}. The final step uses the one-shot decoupling framework developed for random and approximate-design dynamics \cite{DupuisEtAl2014,SzehrEtAl2013,BrownFawzi2015}.

The closest classical parallels are deletion channels and synchronization-error models, where symbol positions are lost during transmission \cite{Mitzenmacher2009,CheraghchiRibeiro2021,WangDumanAktas2013}. Quantum insertion--deletion codes address related synchronization errors by constructing encodings and decoders that tolerate unknown insertions or deletions \cite{LeahyTouchetteYao2019,HagiwaraNakayama2020,ShibayamaHagiwara2021,Ouyang2021,ShibayamaOuyang2021,Hagiwara2023,BulledOuyang2026,OuyangBrennen2026,SasakiEtAl2026}. The present setting differs in that the encoding is generated by scrambling, the input is the entanglement-assisted Hayden--Preskill state, and the objective is recovery of a fixed-size diary. The planted distribution used in the achievability proof also connects the analysis to recent work on random subsequences in classical deletion channels \cite{JeongPernice2026}.

Section~\ref{sec:model} defines the channels and the operational recovery length. Section~\ref{sec:results} states the ordered and block-resolved recovery laws. Section~\ref{sec:proofideas} presents the main ideas behind their proofs. Section~\ref{sec:chronological} studies a permutation-twirled chronological emission channel for comparison. Section~\ref{sec:interpretation} discusses subsystem identity as classical side information. The appendices contain the complete proofs.

\section{Protocol, channels, and operational quantities}
\label{sec:model}

\subsection{Old-black-hole Hayden--Preskill setup}

Let $M$ be a $k$-qubit diary maximally entangled with a reference system $R$, and let $H$ denote the remaining $n-k$ qubits of an old black hole. The early radiation $E$ purifies the maximally mixed state on $H$. The input register is $C=MH$, with reduced state
\begin{equation}
 \rho_{CR}=\Phi_{MR}\otimes\pi_H.
 \label{eq:hp-input}
\end{equation}
The full state on $CRE$ is pure. A unitary $U$ acts on $C$, after which an emission channel produces the accessible output $D$ and, when present, a classical flag $\mathsf F$. The decoder knows $U$ and has access to $D$, $E$, and $\mathsf F$.

Let $\widehat M$ denote the recovered diary register. For fixed $U$ and fixed channel parameters, the optimal entanglement fidelity is
\begin{equation}
 \Fopt(U)
 =\sup_{\mathcal R}
 \langle\Phi_{R\widehat M}|
 (\id_R\otimes\mathcal R)
 (\rho^{U}_{RDE\mathsf F})
 |\Phi_{R\widehat M}\rangle,
 \label{eq:fopt-def}
\end{equation}
where $\mathcal R:DE\mathsf F\to\widehat M$ is a quantum channel. The recovery map may depend on $U$, $n$, $\ell$, and the available side information.

\subsection{Ordered quantum deletion}

For an ordered subset
$S=(S_1<\cdots<S_\ell)\subset[n]$, let $J_S$ denote the canonical order-preserving identification
\begin{equation}
 (\mathbb C^2)_{S_1}\otimes\cdots\otimes(\mathbb C^2)_{S_\ell}
 \longrightarrow
 D_1\otimes\cdots\otimes D_\ell.
\end{equation}
The fixed-length ordered deletion channel is
\begin{equation}
 \D_{n\to\ell}(\rho)
 =\binom n\ell^{-1}
 \sum_{\substack{S\subset[n]\\|S|=\ell}}
 J_S\Tr_{S^c}(\rho)J_S^\dagger.
 \label{eq:ordered-deletion}
\end{equation}
The receiver is not told the surviving set $S$. The output nevertheless preserves the relative order inherited from the microscopic input labels. This is the sense in which the output is unlabeled but ordered.

Let $\tau^{(n,\ell)}_{C'D}$ denote the normalized Choi state of \eqref{eq:ordered-deletion}, where $C'$ is a copy of the input register. Its marginals are maximally mixed:
\begin{equation}
 \tau_{C'}=\pi_{C'},
 \qquad
 \tau_D=\pi_D.
\end{equation}

\subsection{Block-resolved partial labels}

Assume that $n=Bm$ and divide the microscopic positions into $B$ consecutive blocks
$I_1,\ldots,I_B$, each of size $m$. For a surviving set $S$, define
\begin{equation}
 L_b=|S\cap I_b|,
 \qquad
 L=(L_1,\ldots,L_B).
\end{equation}
The receiver learns the block label of each surviving qubit, but not its position within the block. Since the blocks are consecutive and the output retains relative order, this information is equivalent to the count vector $L$ together with the canonical grouping of the output into blocks. The relative microscopic order within each block is preserved.

The corresponding flagged channel is
\begin{equation}
 \D^{(B)}_{n\to\ell}(\rho)
 =\binom n\ell^{-1}
 \sum_{\substack{S\subset[n]\\|S|=\ell}}
 |L(S)\rangle\!\langle L(S)|_{\mathsf L}\otimes
 J_S\Tr_{S^c}(\rho)J_S^\dagger.
 \label{eq:block-channel}
\end{equation}
A count vector $L$ with $\sum_bL_b=\ell$ occurs with probability
\begin{equation}
 p_L=
 \frac{\displaystyle\prod_{b=1}^B\binom m{L_b}}
 {\displaystyle\binom n\ell}.
 \label{eq:count-law}
\end{equation}
Conditioned on $L$, the surviving positions are independent and uniformly distributed within the different blocks. Hence
\begin{equation}
 \left.\D^{(B)}_{n\to\ell}\right|_L
 \simeq
 \bigotimes_{b=1}^B\D_{m\to L_b}.
 \label{eq:conditional-factorization}
\end{equation}
This factorization is the main structural ingredient in the block-resolved phase diagram.

\subsection{Fixed-error operational recovery length}

Choose fixed tolerances
\begin{equation}
 0<\varepsilon<1-4^{-k},
 \qquad
 0<\delta<1.
\end{equation}
The restriction on $\varepsilon$ ensures that the target fidelity lies strictly above the baseline $4^{-k}$.

For the ordered deletion channel, define
\begin{equation}
 \ell_{\rec}(n;\varepsilon,\delta)
 =\min\left\{\ell:
 \Prob_{U\sim\mathrm{Haar}}\!\left[
 \Fopt(U;\ell)\ge1-\varepsilon
 \right]\ge1-\delta\right\}.
 \label{eq:operational-threshold}
\end{equation}
For the block-resolved channel, define similarly
\begin{equation}
 \ell_{\rec}(n,B;\varepsilon,\delta)
 =\min\left\{\ell:
 \Prob_{U\sim\mathrm{Haar}}\!\left[
 \Fopt(U;\ell,B)\ge1-\varepsilon
 \right]\ge1-\delta\right\}.
 \label{eq:operational-threshold-block}
\end{equation}

Discarding one uniformly chosen survivor after applying
$\D_{n\to\ell+1}$ reproduces the channel with $\ell$ survivors:
\begin{equation}
 \D_{n\to\ell}
 =\D_{\ell+1\to\ell}\circ\D_{n\to\ell+1}.
 \label{eq:degradation-single}
\end{equation}
The block-resolved channels satisfy the analogous relation, with the classical count flag updated after the additional deletion. It follows that the optimal fidelity is nondecreasing in $\ell$. This monotonicity converts the subcritical and supercritical asymptotic regimes into constant-factor bounds on the operational recovery length.

\section{Main results}
\label{sec:results}

\subsection{Ordered deletion without microscopic labels}

\begin{theorem}[Ordered-deletion recovery transition]
\label{thm:ordered}
Fix the diary size $k$ and let $n\to\infty$.
\begin{enumerate}[label=(\roman*)]
\item If $\ell=o(n^{2/3})$, then
\begin{equation}
 \Fopt(U;\ell)\longrightarrow 4^{-k}
 \label{eq:ordered-subcritical}
\end{equation}
uniformly over the scrambling unitary $U$.

\item If $\ell/n^{2/3}\to\infty$ and $\ell/n\to0$, then
\begin{equation}
 \Fopt(U;\ell)\longrightarrow 1
 \label{eq:ordered-supercritical}
\end{equation}
in probability over a Haar-random scrambling unitary.
\end{enumerate}
\end{theorem}

The value $4^{-k}$ is achieved by discarding the available systems and preparing the maximally mixed state on the recovered $k$-qubit diary. Since $\Fopt$ is already optimized over all recovery channels, part (i) shows that no decoder asymptotically exceeds this baseline below the $n^{2/3}$ scale.

\begin{corollary}[Operational recovery exponent]
\label{cor:ordered-operational}
For fixed $k,\varepsilon,\delta$ as in
\eqref{eq:operational-threshold}, there exist constants
$0<c_{k,\varepsilon,\delta}\le C_{k,\varepsilon,\delta}<\infty$
such that, for all sufficiently large $n$,
\begin{equation}
 c_{k,\varepsilon,\delta}n^{2/3}
 \le
 \ell_{\rec}(n;\varepsilon,\delta)
 \le
 C_{k,\varepsilon,\delta}n^{2/3}.
 \label{eq:ordered-theta}
\end{equation}
Consequently,
\[
 \ell_{\rec}(n;\varepsilon,\delta)=\Theta(n^{2/3})
\]
for fixed error and failure tolerances.
\end{corollary}

\subsection{Block-resolved microscopic side information}

For a block of $m$ microscopic positions containing $r$ survivors, define
\begin{equation}
 p^{(m,r)}_{ai}
 =
 \frac{\binom{i-1}{a-1}\binom{m-i}{r-a}}
 {\binom mr},
 \qquad
 V_{m,r}
 =
 \sum_{a=1}^r\sum_{i=1}^m
 \left(p^{(m,r)}_{ai}\right)^2,
 \label{eq:block-kernel}
\end{equation}
with $V_{m,0}=0$. Here $p^{(m,r)}_{ai}$ is the probability that the $a$th survivor occupies microscopic position $i$. For a branch $L$, let
\begin{equation}
 X(L)=\sum_{b=1}^B V_{m,L_b}.
 \label{eq:Xbranch}
\end{equation}
The scale governing the transition is
\begin{equation}
 \Lambda_{n,\ell,B}
 =
 \begin{cases}
 B\ell/n, & \ell/B\le1,\\[1mm]
 \sqrt B\,\ell^{3/2}/n, & \ell/B\ge1.
 \end{cases}
 \label{eq:Lambda}
\end{equation}

\begin{theorem}[Block-resolved side-information threshold]
\label{thm:block}
Fix $k$, let $n=Bm$, and assume $\ell=o(n)$.
\begin{enumerate}[label=(\roman*)]
\item If $\Lambda_{n,\ell,B}\to0$, then
\begin{equation}
 \Fopt(U;\ell,B)\longrightarrow 4^{-k}
\end{equation}
uniformly over the scrambling unitary $U$.

\item If $\Lambda_{n,\ell,B}\to\infty$, then
\begin{equation}
 \Fopt(U;\ell,B)\longrightarrow 1
\end{equation}
in probability over a Haar-random scrambling unitary.
\end{enumerate}
\end{theorem}

\begin{corollary}[Side-information phase diagram]
\label{cor:phase}
For fixed $k,\varepsilon,\delta$ and any sequence $B=B(n)$ dividing $n$,
\begin{equation}
 \ell_{\rec}(n,B;\varepsilon,\delta)
 \asymp
 s(n,B),
 \qquad
 s(n,B)=
 \begin{cases}
 n^{2/3}B^{-1/3}, & B\le\sqrt n,\\[1mm]
 n/B, & B\ge\sqrt n.
 \end{cases}
 \label{eq:phase-diagram}
\end{equation}
The implicit constants may depend on $k,\varepsilon,\delta$, but not on $n$ or $B$. The two regimes meet at
\[
 B\asymp\sqrt n,
 \qquad
 \ell_{\rec}\asymp\sqrt n.
\]
\end{corollary}

\begin{figure}[t]
 \centering
 \includegraphics[width=0.62\linewidth]{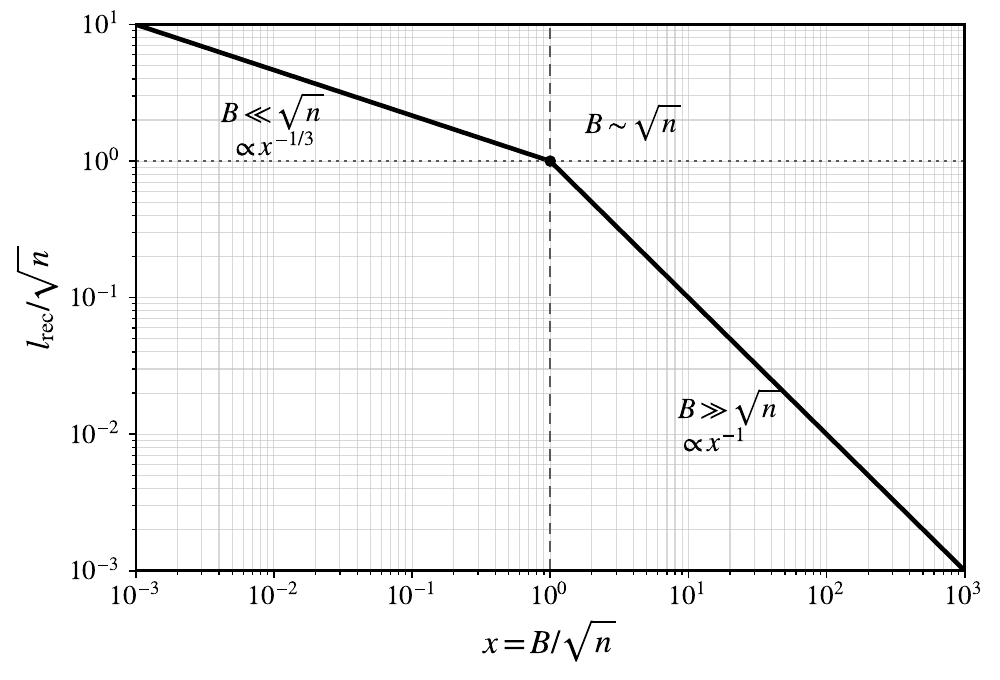}
 \caption{Scaling form of the fixed-error recovery phase diagram. With $x=B/\sqrt n$, one has
 $\ell_{\rec}/\sqrt n\asymp x^{-1/3}$ in the many-survivor-per-block regime and
 $\ell_{\rec}/\sqrt n\asymp x^{-1}$ in the sparse-block regime. The two branches meet at $B\asymp\sqrt n$ and $\ell_{\rec}\asymp\sqrt n$. The plotted branches represent the two asymptotic scaling regimes; the crossover profile for $B/\sqrt n=\Theta(1)$ is not determined by these bounds.}
 \label{fig:phase}
\end{figure}

If $B=2^q$, the block label contains $q$ bits of coarse position information, and
\begin{equation}
 \ell_{\rec}\asymp
 \begin{cases}
 n^{2/3}2^{-q/3},
 & q\le\frac12\log_2n,\\[1mm]
 n2^{-q},
 & q\ge\frac12\log_2n.
 \end{cases}
 \label{eq:qbits}
\end{equation}
Each additional label bit reduces the recovery length by a factor $2^{1/3}$ in the many-survivor-per-block regime and by a factor $2$ in the sparse regime. The phase diagram therefore quantifies the operational value of a concrete side-information resource: exact consecutive block labels together with the relative microscopic order inside each block.

\section{Proof architecture}
\label{sec:proofideas}

\subsection{Exact purity and the universal converse}

Let
\begin{equation}
 S=(S_1<\cdots<S_\ell),
 \qquad
 T=(T_1<\cdots<T_\ell)
\end{equation}
be independent uniformly distributed ordered $\ell$-subsets of $[n]$. Define the number of aligned coincidences by
\begin{equation}
 M_{n,\ell}
 =\sum_{a=1}^{\ell}\bbone_{\{S_a=T_a\}}.
\end{equation}
The normalized Choi state of the ordered deletion channel satisfies
\begin{equation}
 G_{n,\ell}
 :=2^{n+\ell}\Tr\left[(\tau^{(n,\ell)})^2\right]
 =\E 4^{M_{n,\ell}}.
 \label{eq:purity-identity-main}
\end{equation}

To see why the relative order matters, consider an overlap $S_a=T_b$. Since both subsets are increasing, these overlaps define an increasing partial map between their ranks. Such a map has no cycles other than fixed points. The closed loops in the Choi contraction therefore occur only when $a=b$, and each aligned coincidence contributes a factor of four. Nontrivial cycles may appear once the canonical order is removed, which is one of the differences between ordered deletion and the chronological channel considered below.

Writing a uniform subset in terms of its gaps reduces prescribed rank coincidences to collisions between Dirichlet--multinomial block sums. These block sums can be represented by independent negative-binomial variables conditioned on their total. This gives the uniform factorial-moment bound
\begin{equation}
 \frac{\E(M_{n,\ell})_r}{r!}
 \le
 \left[
 C\frac{(\ell+1)^{3/2}}{n}
 \right]^r.
 \label{eq:factorial-main}
\end{equation}
Using
\begin{equation}
 4^M
 =\sum_{r\ge0}3^r\frac{(M)_r}{r!},
\end{equation}
we obtain $G_{n,\ell}=1+o(1)$ whenever
$\ell=o(n^{2/3})$. The Choi state is consequently close to the maximally mixed state:
\begin{equation}
 \frac12\norm{\tau-\pi}_1
 \le
 \frac12\sqrt{G_{n,\ell}-1}.
\end{equation}
Unitary precomposition leaves this trace distance unchanged, while a recovery channel cannot increase it. Comparing with the maximally mixed Choi state gives
\begin{equation}
 \Fopt(U;\ell)
 \le
 4^{-k}
 +\frac12\sqrt{G_{n,\ell}-1}.
 \label{eq:choi-converse-main}
\end{equation}
The bound holds for every scrambling unitary and proves the subcritical part of Theorem~\ref{thm:ordered}.

\subsection{Order-statistic strength and the planted subsequence}

Measuring the Choi state in the computational basis produces the following classical experiment:
\begin{enumerate}[label=(\alph*)]
\item $X\in\{0,1\}^n$ is uniformly distributed;
\item $S$ is a uniformly distributed ordered $\ell$-subset;
\item $Y=X_S$ is the subsequence selected by $S$.
\end{enumerate}
Let
\begin{equation}
 p_{ai}=\Prob(S_a=i),
 \qquad
 V_{n,\ell}=\sum_{a,i}p_{ai}^2.
\end{equation}
The quantity $V_{n,\ell}$ is the total collision strength of the order-statistic kernel. A beta-binomial estimate gives
\begin{equation}
 V_{n,\ell}\asymp\frac{\ell^{3/2}}{n}
 \label{eq:Vscale-main}
\end{equation}
when $\ell\to\infty$ and $\ell=o(n)$. The scale $n^{2/3}$ appears when this quantity becomes of order one.

Write
\begin{equation}
 \xi_i=(-1)^{X_i},
 \qquad
 \eta_a=(-1)^{Y_a},
\end{equation}
and consider the weighted alignment score
\begin{equation}
 T(X,Y)
 =\sum_{a,i}p_{ai}\xi_i\eta_a.
 \label{eq:score-main}
\end{equation}
Under $P_XQ_Y$, where $P_X$ is the uniform marginal of $X$ and $Q_Y$ is arbitrary, this statistic has mean zero. Under the planted law, its conditional mean is
\begin{equation}
 m(S)=\sum_a p_{a,S_a}.
\end{equation}
A martingale estimate based on Freedman's inequality
\cite{Freedman1975} controls the central order statistics of $S$. Combined with a local hypergeometric lower bound, it shows that $m(S)$ is bounded below by $V_{n,\ell}$ up to polylogarithmic factors in $V_{n,\ell}$, with probability tending to one.

For fixed $S$, the score is a Rademacher quadratic form. The associated matrix satisfies
\begin{equation}
 \norm{A_S}_F^2\le V_{n,\ell},
 \qquad
 \norm{A_S}_2^2\le\ell/n.
\end{equation}
The Hanson--Wright inequality
\cite{RudelsonVershynin2013} then shows that the score remains close to its planted mean when $V_{n,\ell}\to\infty$.

Sibson information involves an optimization over the output law $Q_Y$ \cite{Sibson1969}. Accordingly, the separation must hold uniformly in $Q_Y$, not only for the planted output marginal. For a fixed output word $y$, define
\begin{equation}
 v(y)
 =
 \sum_i
 \left(
 \sum_a p_{ai}(-1)^{y_a}
 \right)^2.
\end{equation}
Under the planted output distribution,
\begin{equation}
 \E_{\rm pl} v(Y)=V_{n,\ell}.
\end{equation}
We restrict to outputs for which $v(Y)$ is at most a slowly growing multiple of $V_{n,\ell}$. This restriction has probability tending to one under the planted law. On the same set, Hoeffding's inequality gives an upper bound on the probability of a large score under $P_XQ_Y$, uniformly over $Q_Y$. The resulting event separation forces the Hellinger affinity to vanish uniformly over $Q_Y$ and gives
\begin{equation}
 I_{1/2}^{\rm S}(X:Y)\longrightarrow\infty.
 \label{eq:sibson-main}
\end{equation}

Monotonicity in the R\'enyi order and data processing for sandwiched R\'enyi divergence
\cite{LeditzkyRouzeDatta2016,FrankLieb2013,Beigi2013}
transfer this classical separation to the Choi state:
\begin{equation}
 I_{2/3}^*(C':D)_\tau\longrightarrow\infty.
\end{equation}

For the Hayden--Preskill input in \eqref{eq:hp-input},
\begin{equation}
 H_2^*(C|R)=n-2k.
\end{equation}
Let $C'DF_{\rm env}$ be a pure dilation of the normalized Choi state. Sandwiched R\'enyi duality gives
\begin{equation}
 H_2^*(C'|F_{\rm env})
 =
 -H_{2/3}^*(C'|D)
 =
 I_{2/3}^*(C':D)-n.
\end{equation}
The joint state--channel decoupling theorem of Cheng, Dupuis, and Gao
\cite{ChengDupuisGao2024} therefore bounds the mean complementary-channel error by an exponential in
\begin{equation}
 -\frac12
 \left[
 I_{2/3}^*(C':D)-2k+\log_2 3
 \right].
\end{equation}
Since $k$ is fixed and the R\'enyi information diverges, the mean error tends to zero. Markov's inequality promotes this conclusion to convergence in probability over the Haar-random scrambling unitary. Uhlmann's theorem then provides a recovery channel acting on $DE$.

\subsection{Block factorization with partial labels}

Conditioned on the observed count vector $L$, the block-resolved channel factorizes as in
\eqref{eq:conditional-factorization}. For a single block with $m$ microscopic positions and $r$ survivors, the collision strength satisfies the uniform estimate
\begin{equation}
 V_{m,r}\asymp
 \frac{r^{3/2}}{\sqrt{m(m-r+1)}}.
 \label{eq:uniform-block-V-main}
\end{equation}
Each order statistic follows a beta-binomial law. Its log-concavity makes the collision probability comparable to both its maximum mass and the inverse of its standard deviation, which gives \eqref{eq:uniform-block-V-main}.

For every feasible branch $L$,
\begin{equation}
 X(L)=\sum_bV_{m,L_b}\ge c\Lambda_{n,\ell,B}.
 \label{eq:branch-lower-main}
\end{equation}
When $\ell/B\le1$, each order-statistic row contributes at least $1/m$. When $\ell/B\ge1$, the lower bound follows from the convexity of $r^{3/2}$. In the opposite direction,
\begin{equation}
 \E_LX(L)\le C\Lambda_{n,\ell,B}.
 \label{eq:branch-upper-main}
\end{equation}
Together with the tensor-product form of the conditional Choi state, these estimates reduce the converse to the single-block bound
\eqref{eq:choi-converse-main}.

The achievability argument separates two cases. If
$V_{m,L_b}$ diverges for some block, that block already contains a supercritical ordered-deletion problem. Otherwise, information is distributed among many blocks and the corresponding alignment scores must be combined. The pointwise estimate
\begin{equation}
 q_{m,r}(S)
 =\sum_a p_{a,S_a}^{(m,r)}
 \le C V_{m,r}
\end{equation}
controls the variance of the planted mean. The sum of the block scores then concentrates on the scale $X(L)$. Applying the same variance truncation used in the single-block argument gives divergent Sibson information for each branch. Since the lower bound
\eqref{eq:branch-lower-main} holds deterministically, every feasible branch becomes informative when
$\Lambda_{n,\ell,B}\to\infty$. The classical flag identifies the branch and allows the corresponding recovery channels to be combined.

\section{Chronological emission after permutation twirling}
\label{sec:chronological}

A receiver who knows only the temporal order of randomly selected carriers is described by a different channel from ordered deletion. We model this situation by twirling the input over all permutations and then retaining the first $\ell$ systems:
\begin{equation}
\Ochan_{n\to\ell}
=\Tr_{\ell+1,\ldots,n}\circ\mathcal T_{S_n},
\label{eq:chron-channel}
\end{equation}
where $\mathcal T_{S_n}$ denotes the uniform permutation twirl. The canonical microscopic order used in the ordered-deletion channel is erased before the partial trace.

Let $F_{\bm m}^{(n)}$ be the normalized permutation-symmetric sum of Pauli strings with
\begin{equation}
\bm m=(m_x,m_y,m_z),
\qquad |\bm m|=r.
\end{equation}
The channel is diagonal in this symmetric Pauli basis. On the degree-$r$ sector, its singular value is
\begin{equation}
\lambda_{n,\ell,r}
=\sqrt{\frac{(\ell)_r}{(n)_r}},
\label{eq:chron-singular}
\end{equation}
with multiplicity
\begin{equation}
d_r=\binom{r+2}{2}.
\end{equation}
Here $(a)_r=a(a-1)\cdots(a-r+1)$. If $\tau_{C'D}^{\rm chron}$ denotes the normalized Choi state, then
\begin{equation}
G_{n,\ell}^{\rm chron}
=2^{n+\ell}\Tr[(\tau_{C'D}^{\rm chron})^2]
=\sum_{r=0}^{\ell}
\binom{r+2}{2}\frac{(\ell)_r}{(n)_r}.
\label{eq:chron-purity}
\end{equation}
When $\ell/n\to\alpha<1$, dominated convergence yields
\begin{equation}
G_{n,\ell}^{\rm chron}
\longrightarrow
\frac{1}{(1-\alpha)^3},
\label{eq:chron-limit}
\end{equation}
while each fixed-degree singular value converges to $\alpha^{r/2}$. The coefficient $\binom{r+2}{2}$ counts the degree-$r$ monomials in the three nonidentity Pauli directions, whose generating function produces the cubic denominator.

The chronological channel therefore retains information according to Pauli degree. This differs from ordered deletion, whose purity and recovery transition are governed by rank-aligned coincidences between subsequences. The two channels encode distinct notions of order on the tensor factors. The analysis above determines the spectrum of the chronological channel; establishing diary recovery would additionally require control of its complementary channel.

\section{Interpretation and relation to prior work}
\label{sec:interpretation}

\subsection{Subsystem identity as side information}

The recovery problem depends both on the dimension of the accessible output and on its embedding in the scrambled register. In the labeled Hayden--Preskill protocol, this embedding is part of the channel description, since the collected radiation corresponds to a definite tensor factor. Ordered deletion replaces that definite subsystem by an average over order-preserving embeddings. The averaging leaves a structured output, but it removes enough information to change the recovery length from a constant scale to $n^{2/3}$.

The scale $\ell^{3/2}/n$ measures how much microscopic position information can be inferred collectively from the ordered output. Each output rank gives a weak indication of its microscopic origin, and these indications become informative only after sufficiently many survivors have been collected. Below the transition, the Choi state approaches the maximally mixed state. Above it, the planted alignment becomes statistically detectable. Channel purity controls the converse, while statistical distinguishability supplies the achievability argument; both are governed by the same alignment scale.

Coarse block labels show how the recovery law changes with the resolution of the side information. In a block containing many survivors, the label narrows the order-statistic windows and strengthens their collective alignment. When most occupied blocks contain a single survivor, the block labels nearly identify the microscopic locations. The change at $B\asymp\sqrt n$ occurs when the labels begin to localize individual survivors rather than merely sharpen their collective rank information.

The ordered deletion channel is an abstract information-theoretic model for ordered acquisition without microscopic source labels. It is not derived from a specific Hawking-emission process. Within the Hayden--Preskill setting, it makes explicit the classical information about subsystem identity that is normally taken for granted.

Arrival chronology belongs to a different symmetry class. Permutation twirling erases the microscopic sequence before the carriers are retained, and the surviving information is organized by symmetric Pauli degree. The exact spectrum found in Sec.~\ref{sec:chronological} shows that two notions of order may impose different constraints on the underlying tensor factors and therefore define different quantum channels.

\subsection{Connections with scrambling and decoupling}

Hayden--Preskill recovery can be formulated as a decoupling problem. The diary becomes recoverable from the radiation once it has decoupled from the inaccessible environment. This viewpoint connects the protocol to one-shot coding theorems and to random-circuit constructions of quantum codes \cite{DupuisEtAl2014,SzehrEtAl2013,BrownFawzi2015}. Channel-state duality also relates recovery to scrambling diagnostics based on tripartite information and out-of-time-order correlators \cite{HosurEtAl2016,YoshidaYao2019}. Recent studies of Hamiltonian and circuit dynamics have shown that recovery can distinguish forms of chaos, transport, and symmetry that are not visible in spectral statistics alone \cite{NakataTezuka2024,RamppClaeys2024}. The present results show that the recovery law also depends on how the collected output is classically identified. Even for the same scrambling unitary, changing the available subsystem labels changes the asymptotic amount of output required for recovery.

Existing decoder constructions generally begin after the accessible subsystem has been specified. Yoshida--Kitaev teleportation, Clifford recovery algorithms, and complementarity-based decoders provide explicit procedures in that setting \cite{YoshidaKitaev2017,Yoshida2022,NakataMatsuuraKoashi2025}. The achievability theorem proved here is existential and optimizes over all recovery maps. It determines the output scale on which a decoder can exist after the position labels have been erased.

\subsection{Connections with deletion and synchronization channels}

Classical deletion channels are basic examples of synchronization-error models, and their capacity theory remains substantially more delicate than that of erasure channels \cite{Mitzenmacher2009,CheraghchiRibeiro2021}. Segmented deletion models illustrate how coarse location information can affect synchronization and decoding \cite{WangDumanAktas2013}. Although the input ensemble and decoding objective are different, the block-resolved model likewise shows that location metadata can change the scaling of a recovery problem.

Quantum deletion research has concentrated on code constructions, Knill--Laflamme conditions, insertion--deletion equivalences, and explicit decoding algorithms \cite{LeahyTouchetteYao2019,HagiwaraNakayama2020,ShibayamaHagiwara2021,Ouyang2021,ShibayamaOuyang2021,Hagiwara2023,BulledOuyang2026,OuyangBrennen2026,SasakiEtAl2026}. Those works design code spaces that withstand synchronization errors. In ordered Hayden--Preskill recovery, the encoding is generated by scrambling, and the question is how much unlabeled output is required to reconstruct a fixed logical system. For this encoding, the threshold is governed by the statistics of random subsequences.

The planted distribution used in the achievability proof is closely related to the random-subsequence model studied by Jeong and Pernice in connection with uniform codes for the deletion channel \cite{JeongPernice2026}. Their work concerns a dense regime, whereas the present proof uses $\ell=o(n)$. Here the classical separation is subsequently converted into quantum recovery through R\'enyi data processing and decoupling. This step depends on the channel geometry induced by ordered deletion.

\section{Conclusion}
\label{sec:conclusion}

The microscopic specification of the collected radiation is part of the Hayden--Preskill channel. If $\ell=o(n^{2/3})$, the optimal entanglement fidelity approaches the no-output value uniformly over the scrambling unitary. If $\ell=\omega(n^{2/3})$ with $\ell=o(n)$, a Haar-random scrambling unitary permits asymptotically perfect recovery with probability tending to one. Together with monotonicity, these results establish the fixed-error law $\ell_{\mathrm{rec}}=\Theta(n^{2/3})$.

The exponent $2/3$ arises from the uncertainty in the microscopic positions of the ordered survivors. Correlations accumulate through rank-aligned coincidences, whose total strength is of order $\ell^{3/2}/n$. The exact Choi-purity identity yields the universal converse, while the planted-subsequence argument and joint state--channel decoupling establish achievability on the same scale. The agreement between the two arguments identifies this order-statistic quantity as the mechanism governing the transition.

Block labels show how the recovery law changes when part of the position information is restored. The scale is $n^{2/3}B^{-1/3}$ when a typical block contains several survivors, and it becomes $n/B$ when the block labels nearly identify the individual systems. The crossover at $B\asymp\sqrt n$ marks the change between these two regimes.

Recovery therefore depends jointly on the accessible quantum systems and on the classical information that identifies them within the scrambled register. The ordered-deletion theorem and the block-resolved phase diagram quantify this dependence in the Hayden--Preskill setting.

\appendix

\section{Full proof of the ordered-deletion threshold}
\label{app:ordered-full}

This appendix proves Theorem~\ref{thm:ordered}. Throughout, the surviving tensor factors are identified with the common output register $D$ by the canonical order-preserving map $J_S$.

\subsection{Model and main statement}

The ordered deletion channel is
\begin{equation}
\D_{n\to\ell}(\rho)
=
\binom n\ell^{-1}
\sum_{\substack{S\subset[n]\\ |S|=\ell}}
J_S\Tr_{S^c}(\rho)J_S^\dagger.
\label{eq:deletion}
\end{equation}
Thus the surviving qubits are presented in their original relative order, while the set $S$ is not revealed. Logarithms and entropies are measured in bits unless stated otherwise.

\begin{theorem}[Ordered-deletion Hayden--Preskill threshold]
\label{thm:main}
Fix the diary size $k$ and assume $\ell/n\to0$. For the old-black-hole Hayden--Preskill protocol with the channel \eqref{eq:deletion},
\begin{align}
\ell=o(n^{2/3})
&\quad\Longrightarrow\quad
F_{\rm e}^{\rm opt}(U;\ell)\longrightarrow 4^{-k},
\label{eq:subcrit}\\
\ell=\omega(n^{2/3})
&\quad\Longrightarrow\quad
F_{\rm e}^{\rm opt}(U;\ell)\longrightarrow1
\quad\text{in probability over a Haar-random scrambling unitary.}
\label{eq:supercrit}
\end{align}
The first convergence is uniform over the scrambling unitary $U$. Since $F_{\rm e}^{\rm opt}$ is optimized over all recovery channels, no decoder asymptotically exceeds the value $4^{-k}$ in the subcritical regime.
\end{theorem}

\subsection{Exact Choi purity and the subcritical converse}

Let $\tau_{C'D}^{(n,\ell)}$ be the normalized Choi state of \eqref{eq:deletion}. For independent uniformly distributed ordered $\ell$-subsets
\[
S=(S_1<\cdots<S_\ell),
\qquad
T=(T_1<\cdots<T_\ell),
\]
define
\begin{equation}
M_{n,\ell}
=
\sum_{a=1}^{\ell}\bm 1_{\{S_a=T_a\}}.
\end{equation}

\begin{lemma}[Exact collision identity]
\label{lem:purity}
\begin{equation}
G_{n,\ell}
:=
2^{n+\ell}
\Tr\!\left[(\tau_{C'D}^{(n,\ell)})^2\right]
=
\E\,4^{M_{n,\ell}}.
\label{eq:purity}
\end{equation}
\end{lemma}

\begin{proof}
Write
\[
\tau
=
\binom n\ell^{-1}\sum_S\tau_S,
\qquad
\tau_S
=
\Phi_{C'_S D}\otimes\pi_{C'_{S^c}},
\]
where the identification of the surviving factors with $D$ through $J_S$ is understood.

For fixed $S$ and $T$, the contraction of $\tau_S$ with $\tau_T$ can be represented by a bipartite graph whose left vertices are the output ranks $a\in[\ell]$ and whose right vertices are the microscopic positions $i\in[n]$. An equality $S_a=T_b=i$ induces a partial map $a\mapsto b$ between ranks. Since both subsets are increasing, this map is increasing and therefore has no cycles of length greater than one.

Closed contraction loops occur only at fixed points $a=b$, or equivalently when $S_a=T_a$. In the tensor contraction, paths that remain open contribute the same local factors as in the maximally mixed state. At a fixed point, the two maximally mixed factors that would contribute $1/4$ are replaced by a Bell-state contraction equal to one, producing a relative factor of four. Hence
\[
\Tr(\tau_S\tau_T)
=
2^{-(n+\ell)}4^{M(S,T)}.
\]
Averaging over $S$ and $T$ gives \eqref{eq:purity}. The absence of nontrivial cycles relies on the preservation of relative order and does not hold for the chronological channel.
\end{proof}

\begin{lemma}[Negative-binomial point masses]
\label{lem:nbmasses}
Let
\[
b_{d,p}(z)
=
\binom{z+d-1}{d-1}p^d(1-p)^z,
\qquad
z\in\mathbb Z_{\geq0},
\]
be the probability mass function for the number of failures before the $d$th success. There are absolute constants $0<c<C<\infty$ such that, whenever $0<p\leq1/2$ and $q=1-p$,
\begin{equation}
\sup_{z\geq0}b_{d,p}(z)
\leq
C\frac{p}{\sqrt{dq}}
\qquad(d\geq1).
\label{eq:nbupper}
\end{equation}
Moreover, if $D\geq2$, $N\geq D-1$, and
\[
p=\frac{D-1}{N+D-1},
\qquad
q=\frac{N}{N+D-1},
\]
then
\begin{equation}
b_{D,p}(N)
\geq
c\frac{p}{\sqrt{Dq}}.
\label{eq:nblower}
\end{equation}
\end{lemma}

\begin{proof}
For $d=1$, one has $b_{1,p}(z)=pq^z$, and \eqref{eq:nbupper} follows immediately. Suppose that $d\geq2$. The negative-binomial mass is log-concave and attains its maximum at
\[
z_*=
\left\lfloor\frac{(d-1)q}{p}\right\rfloor.
\]
Set $m=z_*+d-1$ and $k=d-1$. Then
\[
b_{d,p}(z_*)
=
p\,\Prb\{\operatorname{Bin}(m,p)=k\},
\]
and the definition of $z_*$ gives $|k-mp|<1$.

Let $\vartheta=k/m$ and
\[
D(\vartheta\|p)
=
\vartheta\log\frac{\vartheta}{p}
+
(1-\vartheta)\log\frac{1-\vartheta}{q}.
\]
The two-sided Stirling bounds
\[
c\sqrt{s}\,(s/e)^s
\leq
s!
\leq
C\sqrt{s}\,(s/e)^s
\qquad(s\geq1)
\]
imply, for $1\leq k\leq m-1$,
\begin{equation}
\Prb\{\operatorname{Bin}(m,p)=k\}
\leq
\frac{C}{\sqrt{m\vartheta(1-\vartheta)}}
e^{-mD(\vartheta\|p)}.
\label{eq:binomialstirlingupper}
\end{equation}
Since $p\leq1/2$, one has $z_*\geq d-1$, and therefore $1\leq k\leq m-1$. The estimate $|k-mp|<1$ also gives
\[
m\vartheta(1-\vartheta)
\geq
c\,mpq
\geq
c'dq.
\]
The exponential factor in \eqref{eq:binomialstirlingupper} is at most one, so
\[
\Prb\{\operatorname{Bin}(m,p)=k\}
\leq
\frac{C}{\sqrt{dq}}.
\]
This proves \eqref{eq:nbupper}.

For the lower bound, put
\[
m=N+D-1,
\qquad
k=D-1=mp.
\]
Then
\[
b_{D,p}(N)
=
p\,\Prb\{\operatorname{Bin}(m,p)=mp\}.
\]
At the mean, the same Stirling estimates give
\begin{equation}
\Prb\{\operatorname{Bin}(m,p)=mp\}
\geq
\frac{c}{\sqrt{mpq}}.
\label{eq:binomialstirlinglower}
\end{equation}
It follows that
\[
b_{D,p}(N)
\geq
c\frac{p}{\sqrt{(D-1)q}}
\geq
c'\frac{p}{\sqrt{Dq}},
\]
which proves \eqref{eq:nblower}.
\end{proof}

\begin{lemma}[Positive-composition convolution]
\label{lem:composition}
For every $k\geq1$ and $L\geq k$,
\begin{equation}
\sum_{\substack{d_1+\cdots+d_k=L\\ d_j\geq1}}
\frac1{\sqrt{d_1\cdots d_k}}
\leq
C^kL^{k/2-1}
\label{eq:composition}
\end{equation}
for an absolute constant $C$.
\end{lemma}

\begin{proof}
Denote the sum on the left by $A_k(L)$. The case $k=1$ is immediate. Suppose that
\[
A_{k-1}(m)
\leq
C^{k-1}m^{(k-1)/2-1}.
\]
Then
\[
A_k(L)
\leq
C^{k-1}
\sum_{d=1}^{L-k+1}
d^{-1/2}(L-d)^{(k-3)/2}.
\]
Comparison with the corresponding beta integral, with the endpoint terms treated separately, gives
\[
\sum_{d=1}^{L-k+1}
d^{-1/2}(L-d)^{(k-3)/2}
\leq
CL^{k/2-1}
B\left(\frac12,\frac{k-1}{2}\right).
\]
The beta factor is uniformly bounded for $k\geq2$. After enlarging $C$ if necessary, induction proves \eqref{eq:composition}.
\end{proof}

\begin{lemma}[Uniform factorial-moment bound]
\label{lem:factorial}
There is an absolute constant $C<\infty$ such that, if $\ell/n\leq1/2$, then for every $1\leq r\leq\ell$,
\begin{equation}
\frac{\E(M_{n,\ell})_r}{r!}
\leq
\left[
C\frac{(\ell+1)^{3/2}}{n}
\right]^r.
\label{eq:factorial}
\end{equation}
Consequently, if $\ell=o(n^{2/3})$, then
\begin{equation}
\E4^{M_{n,\ell}}
=
1+o(1).
\label{eq:mgfsmall}
\end{equation}
\end{lemma}

\begin{proof}
Set
\[
L=\ell+1,
\qquad
N=n-\ell.
\]
Represent an ordered subset $S=(S_1<\cdots<S_\ell)$ through its gaps:
\[
G_0=S_1-1,
\qquad
G_a=S_{a+1}-S_a-1
\quad(1\leq a<\ell),
\qquad
G_\ell=n-S_\ell.
\]
The vector $G=(G_0,\ldots,G_\ell)$ is uniformly distributed over the weak compositions of $N$ into $L$ parts. Their number is
\[
\binom{N+L-1}{L-1}
=
\binom n\ell.
\]

Fix ranks
\[
1\leq a_1<\cdots<a_r\leq\ell
\]
and define
\[
d_0=a_1,
\qquad
d_j=a_{j+1}-a_j
\quad(1\leq j<r),
\qquad
d_r=L-a_r.
\]
Then $d_j\geq1$ and $d_0+\cdots+d_r=L$. Group the gaps into the corresponding $r+1$ consecutive blocks and let $Z=(Z_0,\ldots,Z_r)$ denote their sums. Its probability mass function is
\begin{equation}
\Prb\{Z=z\}
=
\binom n\ell^{-1}
\prod_{j=0}^{r}
\binom{z_j+d_j-1}{d_j-1},
\qquad
z_j\geq0,
\quad
\sum_jz_j=N.
\label{eq:DMlaw}
\end{equation}

For two independent subsets $S$ and $T$, the conditions
\[
S_{a_j}=T_{a_j}
\qquad(1\leq j\leq r)
\]
are equivalent to equality of the first $r$ cumulative block sums. Since the total gap sum is $N$ for both subsets, this is equivalent to equality of all $r+1$ block sums. Therefore
\begin{equation}
\Prb\{S_{a_j}=T_{a_j}\ \forall j\}
=
\sum_z\Prb\{Z=z\}^2
\leq
\max_z\Prb\{Z=z\}.
\label{eq:collisionmax}
\end{equation}

To bound the maximum, set
\[
p=\frac{\ell}{n}
=
\frac{L-1}{N+L-1},
\qquad
q=1-p
=
\frac{N}{n}.
\]
Let $Y_0,\ldots,Y_r$ be independent negative-binomial variables with shapes $d_0,\ldots,d_r$ and common success probability $p$. Their sum is negative-binomial with shape $L$, and conditioning on
\[
\sum_jY_j=N
\]
gives the law in \eqref{eq:DMlaw}. Hence, for every admissible $z$,
\begin{align}
\Prb\{Z=z\}
&=
\frac{\prod_{j=0}^{r}b_{d_j,p}(z_j)}
{b_{L,p}(N)}
\\
&\leq
C^{r+1}
\frac{p^{r+1}q^{-(r+1)/2}}
{\sqrt{d_0\cdots d_r}}
\frac{\sqrt{Lq}}{p}
\\
&\leq
C^r
\left(\frac{L}{n}\right)^r
\frac{\sqrt L}{\sqrt{d_0\cdots d_r}}.
\label{eq:rankmatch}
\end{align}
The second line follows from Lemma~\ref{lem:nbmasses}. In the last line we used $q\geq1/2$ and $p\leq L/n$.

Since
\[
\frac{(M_{n,\ell})_r}{r!}
=
\sum_{1\leq a_1<\cdots<a_r\leq\ell}
\prod_{j=1}^{r}
\bm1_{\{S_{a_j}=T_{a_j}\}},
\]
taking expectations and applying \eqref{eq:rankmatch} gives
\begin{align}
\frac{\E(M_{n,\ell})_r}{r!}
&\leq
C^r
\left(\frac{L}{n}\right)^r
\sqrt L
\sum_{\substack{d_0+\cdots+d_r=L\\d_j\geq1}}
(d_0\cdots d_r)^{-1/2}
\\
&\leq
C^r
\left(\frac{L}{n}\right)^r
\sqrt L\,
C^{r+1}L^{(r+1)/2-1}
\\
&\leq
\left(
C\frac{L^{3/2}}{n}
\right)^r.
\end{align}
This proves \eqref{eq:factorial}, after adjusting the absolute constant.

For every nonnegative integer-valued random variable $M$ and every $q\geq1$,
\[
q^M
=
\sum_{r\geq0}
(q-1)^r\frac{(M)_r}{r!}.
\]
With
\[
x_n=
C\frac{L^{3/2}}{n},
\]
we therefore have
\[
1
\leq
\E4^{M_{n,\ell}}
\leq
\sum_{r=0}^{\ell}(3x_n)^r.
\]
If $\ell=o(n^{2/3})$, then $x_n\to0$, and the geometric sum tends to one. This proves \eqref{eq:mgfsmall}.
\end{proof}

\begin{proposition}[Universal converse]
\label{prop:converse}
For every scrambling unitary $U$,
\begin{equation}
F_{\rm e}^{\rm opt}(U;\ell)
\leq
4^{-k}
+
\frac12\sqrt{G_{n,\ell}-1}.
\label{eq:convbound}
\end{equation}
Consequently, \eqref{eq:subcrit} holds.
\end{proposition}

\begin{proof}
Both marginals of the normalized Choi state are maximally mixed. With
\[
\pi_{C'D}
=
2^{-(n+\ell)}I_{C'D},
\]
the relation between trace and Hilbert--Schmidt norms gives
\[
\frac12\norm{\tau-\pi_{C'D}}_1
\leq
\frac12
\sqrt{
2^{n+\ell}\Tr(\tau^2)-1
}
=
\frac12\sqrt{G_{n,\ell}-1}.
\]

For the old-black-hole input, the pure state
\[
\Phi_{MR}\otimes\Phi_{HE}
\]
is maximally entangled between $C=MH$ and a reference register identified with $C'\simeq RE$, up to a fixed reordering isometry. Applying the scrambling unitary and the deletion channel to $C$ therefore produces the corresponding Choi output, up to a unitary transformation on the reference. Such a transformation preserves trace distance. Any recovery channel acting on $DE$ can only decrease it.

For the product state $\pi_{C'}\otimes\pi_D$, the diary reference $R$ is independent of the decoder input. Its overlap with a maximally entangled state on $R\widehat M$ is therefore exactly $4^{-k}$. The difference between the two overlaps with $\Phi_{R\widehat M}$ is bounded by half the trace distance of the corresponding recovered states. This yields \eqref{eq:convbound}.

Conversely, a decoder can always discard its input and prepare the maximally mixed state on $\widehat M$, achieving fidelity $4^{-k}$. Lemma~\ref{lem:factorial} then proves the subcritical limit.
\end{proof}

\subsection{The planted-subsequence statistic}

Measuring $C'$ and $D$ in the computational basis gives a classical planted-subsequence model. Under the planted law $P$, the input word $X\in\{0,1\}^n$ is uniform, $S$ is a uniformly chosen ordered $\ell$-subset of $[n]$, and
\[
Y_a=X_{S_a}.
\]
Since $X$ is uniform, the marginal distribution of $Y$ is also uniform. Define
\begin{equation}
p_{ai}
=
\Prb(S_a=i)
=
\frac{\binom{i-1}{a-1}\binom{n-i}{\ell-a}}
{\binom n\ell},
\qquad
P_0=(p_{ai})_{a,i},
\end{equation}
and
\begin{equation}
V_{n,\ell}
=
\norm{P_0}_F^2
=
\sum_{a,i}p_{ai}^2.
\end{equation}

\begin{lemma}[Kernel scale and norm]
\label{lem:kernel}
There are absolute constants $0<c<C<\infty$ such that, whenever
$\ell\to\infty$ and $\ell/n\to0$,
\begin{equation}
c\frac{\ell^{3/2}}{n}
\leq
V_{n,\ell}
\leq
C\frac{\ell^{3/2}}{n}.
\label{eq:Vscale}
\end{equation}
Moreover,
\begin{equation}
\norm{P_0}_2^2
=
\frac{\ell}{n}.
\label{eq:Pnorm}
\end{equation}
\end{lemma}

\begin{proof}
The rows of $P_0$ sum to one, and its columns sum to $\ell/n$. Hence
\[
\norm{P_0}_2
\leq
\sqrt{\norm{P_0}_1\norm{P_0}_\infty}
=
\sqrt{\ell/n}.
\]
The normalized constant vectors attain this bound, proving
\eqref{eq:Pnorm}.

To estimate the Frobenius norm, put
\[
L=\ell+1,
\qquad
N=n-\ell.
\]
The shifted order statistic
\[
Z_a:=S_a-a
\]
has the beta-binomial mass
\begin{equation}
\Prb\{Z_a=z\}
=
\binom n\ell^{-1}
\binom{z+a-1}{a-1}
\binom{N-z+L-a-1}{L-a-1},
\qquad
0\leq z\leq N.
\label{eq:betabinomial}
\end{equation}
Equivalently, $Z_a$ is the first component of a two-block
Dirichlet--multinomial distribution with block sizes $a$ and $L-a$.
Its variance is
\begin{equation}
\sigma_a^2
=
\frac{N\,a(L-a)(N+L)}
{L^2(L+1)}.
\label{eq:ordervar}
\end{equation}
Since $\ell/n\to0$, uniformly for $1\leq a\leq\ell$,
\begin{equation}
\sigma_a
\asymp
\frac{n\sqrt{a(L-a)}}{L^{3/2}},
\qquad
\sigma_a\longrightarrow\infty.
\label{eq:sigmaasymp}
\end{equation}

We first bound the maximum point mass. In the negative-binomial
representation used in Lemma~\ref{lem:factorial}, the two block sums
have shapes $a$ and $L-a$ and are conditioned to have total $N$.
Lemma~\ref{lem:nbmasses} gives, uniformly in $z$,
\begin{align}
\Prb\{Z_a=z\}
&\leq
\frac{
C p(aq)^{-1/2}\,
C p((L-a)q)^{-1/2}
}{
c p(Lq)^{-1/2}
}
\\
&\leq
C\frac{L^{3/2}}
{n\sqrt{a(L-a)}},
\label{eq:rowmax}
\end{align}
where $p=\ell/n$ and $q=1-p$. Therefore
\begin{equation}
\sum_i p_{ai}^2
\leq
\max_i p_{ai}
\leq
C\frac{L^{3/2}}
{n\sqrt{a(L-a)}}.
\label{eq:rowupper}
\end{equation}

For the lower bound, Chebyshev's inequality and
\eqref{eq:ordervar} give
\[
\Prb\{|Z_a-\E Z_a|\leq2\sigma_a\}
\geq
\frac34.
\]
The interval contains at most $4\sigma_a+2\leq C\sigma_a$ integers.
Cauchy--Schwarz then yields
\begin{equation}
\sum_i p_{ai}^2
\geq
\frac{(3/4)^2}{4\sigma_a+2}
\geq
c\frac{L^{3/2}}
{n\sqrt{a(L-a)}}.
\label{eq:rowlower}
\end{equation}

Finally,
\[
c
\leq
\sum_{a=1}^{L-1}
\frac1{\sqrt{a(L-a)}}
\leq
C.
\]
The upper bound follows by comparison with
\[
\int_0^1\frac{du}{\sqrt{u(1-u)}},
\]
while the lower bound follows by restricting the sum to
$a\in[L/4,3L/4]$. Summing \eqref{eq:rowupper} and
\eqref{eq:rowlower} over $a$ proves \eqref{eq:Vscale}.
\end{proof}

Set
\[
\xi_i=(-1)^{X_i},
\qquad
\eta_a=(-1)^{Y_a},
\]
and define
\begin{equation}
T(X,Y)
=
\sum_{a=1}^{\ell}\sum_{i=1}^{n}
p_{ai}\xi_i\eta_a.
\label{eq:score}
\end{equation}
Conditioned on the planted subset $S$,
\begin{equation}
\E[T\mid S]
=
m(S)
:=
\sum_{a=1}^{\ell}p_{a,S_a}.
\label{eq:mS}
\end{equation}

\begin{lemma}[Hypergeometric-bridge maximal bound]
\label{lem:bridge}
Fix $\delta\in(0,1/2)$. Let
\[
N_t=|S\cap[t]|,
\qquad
p=\ell/n.
\]
For $1\leq u=o(\sqrt\ell)$,
\begin{equation}
\Prb\left[
\max_{t\leq(1-\delta)n}
|N_t-pt|
>
u\sqrt\ell
\right]
\leq
2e^{-c_\delta u^2}.
\label{eq:bridge}
\end{equation}
Consequently, after changing $c_\delta$ if necessary, with probability at least
$1-2e^{-c_\delta u^2}$,
\begin{equation}
\max_{\delta\ell\leq a\leq(1-\delta)\ell}
\left|
\frac{S_a}{n}
-
\frac{a}{\ell}
\right|
\leq
\frac{u}{\sqrt\ell}.
\label{eq:quantile}
\end{equation}
\end{lemma}

\begin{proof}
Let
\[
\mathcal F_t
=
\sigma\!\left(
\bm1_{\{1\in S\}},\ldots,\bm1_{\{t\in S\}}
\right),
\]
and define
\[
Z_t=N_t-pt,
\qquad
M_t=\frac{Z_t}{n-t},
\qquad
0\leq t<n.
\]
The conditional probability of selecting position $t+1$ is
\[
q_t
=
\Prb(t+1\in S\mid\mathcal F_t)
=
\frac{\ell-N_t}{n-t}.
\]
Since
\[
q_t-p
=
-\frac{Z_t}{n-t},
\]
we have
\[
\E(M_{t+1}\mid\mathcal F_t)=M_t
\]
and
\[
M_{t+1}-M_t
=
\frac{\bm1_{\{t+1\in S\}}-q_t}
{n-t-1}.
\]
Thus $(M_t)$ is a martingale.

Let
\[
\tau
=
\inf\{t\geq0:|Z_t|\geq u\sqrt\ell\},
\]
and stop the process at
$\tau\wedge\lfloor(1-\delta)n\rfloor$.
Before the stopping time,
\[
q_t
\leq
p+\frac{u\sqrt\ell}{\delta n}
=
p\left(
1+\frac{u}{\delta\sqrt\ell}
\right)
\leq
2p
\]
for all sufficiently large $n$. The stopped increments are bounded by
\[
b_\delta
=
\frac{2}{\delta n},
\]
and the predictable quadratic variation is at most
\begin{align*}
\sum_{t<\tau\wedge(1-\delta)n}
\E[(M_{t+1}-M_t)^2\mid\mathcal F_t]
&\leq
2p
\sum_{t\leq(1-\delta)n}
\frac1{(n-t-1)^2}
\\
&\leq
C_\delta\frac{\ell}{n^2}
=:v_\delta.
\end{align*}

If $\tau\leq(1-\delta)n$, then
\[
|M_\tau|
\geq
\frac{u\sqrt\ell}{n}.
\]
Applying Freedman's inequality to $M$ and $-M$ with
\[
x=\frac{u\sqrt\ell}{n}
\]
gives
\[
\Prb(\tau\leq(1-\delta)n)
\leq
2\exp\left[
-\frac{x^2}
{2(v_\delta+b_\delta x/3)}
\right]
\leq
2e^{-c_\delta u^2}.
\]
This proves \eqref{eq:bridge}.

For \eqref{eq:quantile}, apply the preceding bound with $\delta/2$.
On the resulting event, for
$\delta\ell\leq a\leq(1-\delta)\ell$ and sufficiently large $n$,
\[
N_{\lfloor(1-\delta/2)n\rfloor}
\geq
(1-\delta/2)\ell-u\sqrt\ell
>
a.
\]
Hence $S_a\leq(1-\delta/2)n$. Since $N_{S_a}=a$,
\[
|a-pS_a|
\leq
u\sqrt\ell.
\]
Dividing by $\ell$ proves \eqref{eq:quantile}.
\end{proof}

\begin{lemma}[Uniform local hypergeometric lower bound]
\label{lem:localhyp}
Fix $\delta\in(0,1/2)$. There are constants
$c_\delta,C_\delta>0$ with the following property. Let
\[
H\sim\operatorname{Hyp}(N,K,m),
\qquad
\theta=K/N,
\]
and assume
\[
\theta\in[\delta,1-\delta],
\qquad
m\leq N/2,
\qquad
m\longrightarrow\infty.
\]
If $1\leq u=o(\sqrt m)$ and the integer $r$ satisfies
\begin{equation}
|r-m\theta|
\leq
u\sqrt m,
\label{eq:hypwindow}
\end{equation}
then, for all sufficiently large $N$,
\begin{equation}
\Prb(H=r)
\geq
\frac{c_\delta}{\sqrt m}
e^{-C_\delta u^2}.
\label{eq:localhyp}
\end{equation}
Consequently, suppose that $\ell\to\infty$ and $\ell/n\to0$. If
$\delta\ell\leq a\leq(1-\delta)\ell$ and
\begin{equation}
\left|
\frac{i}{n}
-
\frac{a}{\ell}
\right|
\leq
\frac{u}{\sqrt\ell},
\qquad
u=o(\sqrt\ell),
\label{eq:orderwindow}
\end{equation}
then, uniformly over such $a$ and $i$,
\begin{equation}
p_{ai}
\geq
c_\delta
\frac{\sqrt\ell}{n}
e^{-C_\delta u^2}.
\label{eq:pailower}
\end{equation}
\end{lemma}

\begin{proof}
Write
\[
h(k)=\Prb(H=k).
\]
The hypergeometric mass function is log-concave and hence unimodal. Its variance satisfies
\[
\operatorname{Var}(H)
=
m\theta(1-\theta)
\frac{N-m}{N-1}
\leq
\frac m4.
\]
Chebyshev's inequality gives
\[
\Prb\bigl(|H-m\theta|\leq\sqrt m\bigr)
\geq
\frac34.
\]
Since this interval contains at most $2\sqrt m+2$ integers, any mode $r_0$ satisfies
\begin{equation}
h(r_0)
=
\max_k h(k)
\geq
\frac{c}{\sqrt m}.
\label{eq:modemass}
\end{equation}

The ratio of consecutive masses is
\begin{equation}
\frac{h(k+1)}{h(k)}
=
\frac{(K-k)(m-k)}
{(k+1)(N-K-m+k+1)}.
\label{eq:hypratio}
\end{equation}
A mode may be chosen as
\[
r_0
=
\left\lfloor
\frac{(m+1)(K+1)}{N+2}
\right\rfloor,
\]
with the adjacent integer also a mode in the tie case. Moreover,
\[
\left|r_0-m\theta\right|
\leq
3.
\]

Put $k=m\theta+s$. After cancelling the constant factors in
\eqref{eq:hypratio},
\[
\frac{h(k+1)}{h(k)}
=
\frac{
\left(1-\frac{s}{\theta(N-m)}\right)
\left(1-\frac{s}{(1-\theta)m}\right)
}{
\left(1+\frac{s+1}{\theta m}\right)
\left(1+\frac{s+1}{(1-\theta)(N-m)}\right)
}.
\]
If $|s|+1=o(m)$, all four factors remain uniformly bounded away from zero and infinity. Therefore
\begin{equation}
\left|
\log\frac{h(k+1)}{h(k)}
\right|
\leq
C_\delta\frac{|s|+1}{m}.
\label{eq:ratiobound}
\end{equation}

Every integer between $r_0$ and $r$ satisfies
\[
|k-m\theta|
\leq
u\sqrt m+O(1)
=
o(m).
\]
Summing \eqref{eq:ratiobound} from $r_0$ to $r$ gives
\[
\left|
\log\frac{h(r)}{h(r_0)}
\right|
\leq
C_\delta
\left(
\frac{|r-r_0|^2}{m}+1
\right)
\leq
C_\delta(u^2+1).
\]
Combining this estimate with \eqref{eq:modemass} proves
\eqref{eq:localhyp}, after adjusting $c_\delta$.

For the order-statistic consequence, use
\[
p_{ai}
=
\frac{\ell}{n}
\Prb\!\left[
\operatorname{Hyp}(n-1,i-1,\ell-1)
=
a-1
\right].
\]
Set
\[
N=n-1,
\qquad
K=i-1,
\qquad
m=\ell-1,
\qquad
r=a-1.
\]
Under \eqref{eq:orderwindow}, the centrality of $a$ and the condition
$u=o(\sqrt\ell)$ imply, uniformly,
\[
\frac{K}{N}
\in
[\delta/2,1-\delta/2]
\]
and
\[
\left|
r-m\frac{K}{N}
\right|
\leq
C_\delta u\sqrt\ell.
\]
Applying \eqref{eq:localhyp} with adjusted constants proves
\eqref{eq:pailower}.
\end{proof}

\begin{lemma}[Planted score lower tail]
\label{lem:plantedscore}
Assume that
\[
V_{n,\ell}\to\infty,
\qquad
\ell/n\to0.
\]
There is a deterministic sequence $a_n$ such that
\begin{equation}
a_n\to\infty,
\qquad
\frac{a_n^2}{V_{n,\ell}}\to\infty,
\qquad
\frac{a_n}{\sqrt{\ell/n}}\to\infty,
\label{eq:anprops}
\end{equation}
and
\begin{equation}
P\left(
T<\frac{a_n}{2}
\right)
=
o(1).
\label{eq:planttail}
\end{equation}
One may take
\begin{equation}
a_n
=
c\frac{V_{n,\ell}}
{[\log(e+V_{n,\ell})]^K}
\label{eq:an}
\end{equation}
for suitable constants $c,K>0$.
\end{lemma}

\begin{proof}
Fix $\delta\in(0,1/4)$ and choose
\[
u_n^2
=
c_0\log\log(e+V_{n,\ell})
\]
with $c_0>0$ fixed. Since
\[
V_{n,\ell}\to\infty
\qquad\text{and}\qquad
V_{n,\ell}=O(\sqrt\ell),
\]
we have
\[
u_n\to\infty,
\qquad
u_n=o(\sqrt\ell).
\]
Lemma~\ref{lem:bridge} shows that, with probability $1-o(1)$,
\[
\left|
\frac{S_a}{n}
-
\frac{a}{\ell}
\right|
\leq
\frac{u_n}{\sqrt\ell}
\]
simultaneously for all
$a\in[\delta\ell,(1-\delta)\ell]$.
Lemma~\ref{lem:localhyp} then gives
\[
p_{a,S_a}
\geq
c_\delta
\frac{\sqrt\ell}{n}
e^{-C_\delta u_n^2}
\]
for every such rank. Summing over the central ranks yields
\begin{align}
m(S)
&\geq
c_\delta
\frac{\ell^{3/2}}{n}
e^{-C_\delta u_n^2}
\\
&\geq
c\frac{V_{n,\ell}}
{[\log(e+V_{n,\ell})]^K}
=:a_n,
\label{eq:mSlower}
\end{align}
where Lemma~\ref{lem:kernel} was used in the second line.

The first two limits in \eqref{eq:anprops} follow from
\[
\frac{V}{[\log(e+V)]^K}\to\infty,
\qquad
\frac{V}{[\log(e+V)]^{2K}}\to\infty.
\]
For the third, Lemma~\ref{lem:kernel} gives
\[
\frac{a_n}{\sqrt{\ell/n}}
\asymp
\frac{\ell/\sqrt n}
{[\log(e+V_{n,\ell})]^K}.
\]
Since $V_{n,\ell}\to\infty$ is equivalent to
$\ell/n^{2/3}\to\infty$ in the present regime, write
\[
\ell=n^{2/3}g_n,
\qquad
g_n\to\infty.
\]
The last expression is then
\[
\frac{n^{1/6}g_n}
{[\log(e+V_{n,\ell})]^K}.
\]
Since $V_{n,\ell}=O(\sqrt\ell)=O(\sqrt n)$, one has
$\log(e+V_{n,\ell})=O(\log n)$. Therefore
$n^{1/6}/(\log n)^K\to\infty$, and the displayed expression tends to infinity.

It remains to control the fluctuations of $T$ about its conditional mean. Let $R_S$ be the $\ell\times n$ selector matrix whose $a$th row is $e_{S_a}^T$, and define
\[
B_S=P_0^TR_S,
\qquad
A_S=\frac{B_S+B_S^T}{2}.
\]
Then
\[
T=\xi^TA_S\xi,
\qquad
\Tr A_S=m(S).
\]
Since symmetrization does not increase either the Frobenius or operator norm, and
$R_SR_S^T=I_\ell$,
\[
\norm{A_S}_F^2
\leq
\norm{B_S}_F^2
=
\norm{P_0}_F^2
=
V_{n,\ell},
\]
and
\[
\norm{A_S}_2
\leq
\norm{B_S}_2
\leq
\norm{P_0}_2\norm{R_S}_2
=
\sqrt{\ell/n}.
\]

On the event \eqref{eq:mSlower}, the Hanson--Wright inequality gives
\[
\Prb\left(
T<\frac{a_n}{2}
\,\middle|\,
S
\right)
\leq
2\exp\left[
-c\min\left\{
\frac{a_n^2}{V_{n,\ell}},
\frac{a_n}{\sqrt{\ell/n}}
\right\}
\right]
=
o(1).
\]
Adding the $o(1)$ probability that \eqref{eq:mSlower} fails proves
\eqref{eq:planttail}.
\end{proof}

\subsection{Uniform product separation and Sibson information}

Sibson information involves an optimization over all distributions $Q_Y$. It is therefore necessary to distinguish the planted law uniformly from every product law $P_XQ_Y$, where $P_X$ is the fixed uniform input marginal.

For $y\in\{0,1\}^{\ell}$, define
\begin{equation}
c_i(y)
=
\sum_a p_{ai}(-1)^{y_a},
\qquad
v(y)
=
\sum_i c_i(y)^2
=
\norm{P_0^T\eta(y)}_2^2.
\label{eq:vy}
\end{equation}

\begin{lemma}[Uniform separation from product laws]
\label{lem:uniformproduct}
Under the assumptions of Lemma~\ref{lem:plantedscore}, there are events
$\mathcal A_n$ such that
\begin{equation}
P(\mathcal A_n)\longrightarrow1,
\qquad
\sup_{Q_Y}(P_XQ_Y)(\mathcal A_n)\longrightarrow0.
\label{eq:uniformsep}
\end{equation}
\end{lemma}

\begin{proof}
The planted marginal of $Y$ is uniform. Hence
\begin{equation}
\E_P v(Y)
=
\E_\eta\eta^TP_0P_0^T\eta
=
\Tr(P_0P_0^T)
=
V_{n,\ell}.
\end{equation}
Let
\[
h_n
=
\frac{a_n}{\sqrt{V_{n,\ell}}}.
\]
By Lemma~\ref{lem:plantedscore}, $h_n\to\infty$. Define
\[
\mathcal A_n
=
\left\{
T\geq\frac{a_n}{2}
\right\}
\cap
\left\{
v(Y)\leq h_nV_{n,\ell}
\right\}.
\]
Lemma~\ref{lem:plantedscore} and Markov's inequality give
\[
P(\mathcal A_n^c)
\leq
o(1)+h_n^{-1}
=
o(1).
\]

Now fix a distribution $Q_Y$. Under $P_XQ_Y$, conditionally on
$Y=y$,
\[
T
=
\sum_i c_i(y)\xi_i
\]
is a centered Rademacher sum. Whenever
$v(y)\leq h_nV_{n,\ell}$, Hoeffding's inequality gives
\[
(P_XQ_Y)
\left(
T\geq\frac{a_n}{2}
\,\middle|\,
Y=y
\right)
\leq
\exp\left[
-\frac{a_n^2}
{8h_nV_{n,\ell}}
\right]
=
e^{-h_n/8}.
\]
The bound is independent of $Q_Y$, which proves
\eqref{eq:uniformsep}.
\end{proof}

For classical probability distributions, let
\[
\operatorname{Aff}(P,R)
=
\sum_z\sqrt{P(z)R(z)}
\]
denote the Hellinger affinity. For the quantum quantities, we follow the conventions of Ref.~\cite{ChengDupuisGao2024}:
\begin{align}
H_\alpha^*(A|B)_\rho
&=
-\inf_{\sigma_B}
D_\alpha^*
\left(
\rho_{AB}
\middle\|
I_A\otimes\sigma_B
\right),
\\
I_\alpha^*(A:B)_\rho
&=
\inf_{\sigma_B}
D_\alpha^*
\left(
\rho_{AB}
\middle\|
\rho_A\otimes\sigma_B
\right).
\label{eq:starconventions}
\end{align}
Here $D_\alpha^*$ is the sandwiched R\'enyi divergence.

\begin{proposition}[Divergence of the optimized order-$1/2$ information]
\label{prop:sibson}
If
\[
\ell/n\to0,
\qquad
\ell/n^{2/3}\to\infty,
\]
then
\begin{equation}
I_{1/2}^{\rm S}(X:Y)_P
\longrightarrow
\infty.
\label{eq:sibsonhalf}
\end{equation}
Consequently,
\begin{equation}
I_{2/3}^{*}(C':D)_\tau
\longrightarrow
\infty.
\label{eq:qmi}
\end{equation}
\end{proposition}

\begin{proof}
For any event $\mathcal A$,
\[
\operatorname{Aff}(P,R)
\leq
\sqrt{P(\mathcal A)R(\mathcal A)}
+
\sqrt{P(\mathcal A^c)R(\mathcal A^c)}.
\]
Applying Lemma~\ref{lem:uniformproduct} with
$R=P_XQ_Y$ gives
\[
\sup_{Q_Y}
\operatorname{Aff}(P_{XY},P_XQ_Y)
\longrightarrow0.
\]
Since
\[
D_{1/2}(P\|R)
=
-2\log_2\operatorname{Aff}(P,R),
\]
it follows that
\[
I_{1/2}^{\rm S}(X:Y)_P
:=
\inf_{Q_Y}
D_{1/2}(P_{XY}\|P_XQ_Y)
\longrightarrow
\infty.
\]

Monotonicity in the R\'enyi order gives
\[
I_{2/3}^{\rm S}(X:Y)_P
\geq
I_{1/2}^{\rm S}(X:Y)_P.
\]
Computational-basis measurement maps the Choi state to $P_{XY}$. Every classical law $Q_Y$ is obtained by measuring a diagonal conditioning state on $D$. Data processing for the sandwiched R\'enyi divergence, followed by optimization over the conditioning state, therefore gives
\[
I_{2/3}^{*}(C':D)_\tau
\geq
I_{2/3}^{\rm S}(X:Y)_P.
\]
This proves \eqref{eq:qmi}.
\end{proof}

\subsection{Haar recovery by joint state--channel decoupling}

Let $C=MH$ be the $n$-qubit input. The $k$-qubit diary $M$ is maximally entangled with $R$, while the remaining register $H$ is maximally mixed after its purifier $E$ is traced out. Thus
\[
\rho_{CR}
=
\Phi_{MR}\otimes\pi_H,
\qquad
H_2^{*}(C|R)_\rho
=
n-2k.
\]

Let
\[
\D^c_{n\to\ell}:C\to F_{\rm env}
\]
be a complementary channel, and let
$\omega_{C'DF_{\rm env}}$ be its pure normalized Choi state. Pure-state duality for sandwiched conditional entropy gives
\[
H_2^{*}(C'|F_{\rm env})_\omega
=
-H_{2/3}^{*}(C'|D)_\omega.
\]
Since $\omega_{C'}=\pi_{C'}$,
\[
H_{2/3}^{*}(C'|D)_\omega
=
n-I_{2/3}^{*}(C':D)_\tau.
\]
Therefore
\begin{equation}
H_2^{*}(C|R)_\rho
+
H_2^{*}(C'|F_{\rm env})_\omega
=
I_{2/3}^{*}(C':D)_\tau-2k.
\label{eq:balance}
\end{equation}

\begin{proposition}[Supercritical Haar recovery]
\label{prop:achieve}
If
\[
\ell/n\to0,
\qquad
\ell/n^{2/3}\to\infty,
\]
then the optimal entanglement fidelity tends to one in probability over a Haar-random scrambling unitary.
\end{proposition}

\begin{proof}
Apply Theorem~2 of Ref.~\cite{ChengDupuisGao2024} at R\'enyi order $2$ to the state $\rho_{CR}$ and the complementary channel
$\D^c:C\to F_{\rm env}$. The theorem is stated using natural logarithms and a normalized Choi state. In the present base-two convention, it gives
\[
\frac12
\E_U
\norm{
(\D^c\otimes\id_R)
\left(
U\rho_{CR}U^\dagger
\right)
-
\omega_{F_{\rm env}}\otimes\rho_R
}_1
\leq
2^{-\frac12[
H_2^{*}(C|R)_\rho+
H_2^{*}(C'|F_{\rm env})_\omega+
\log_2 3]}.
\]
By \eqref{eq:balance} and Proposition~\ref{prop:sibson}, the right-hand side tends to zero for fixed $k$. Markov's inequality then gives decoupling in probability over the Haar-random scrambling unitary.

Restoring the purifier $E$ of $H$, Uhlmann's theorem, or equivalently decoupling--recovery duality, provides a recovery channel acting on $DE$ whose entanglement fidelity tends to one.
\end{proof}

\subsection{Operational fixed-error threshold}

The degradation identity
\[
\D_{n\to\ell}
=
\D_{\ell+1\to\ell}
\circ
\D_{n\to\ell+1}
\]
shows that $\Fopt(U;\ell)$ is nondecreasing in $\ell$.

Fix $\varepsilon$ and $\delta$ as in
\eqref{eq:operational-threshold}. Suppose first that no positive constant lower bound in
\eqref{eq:ordered-theta} existed. There would then be a sequence $n_j$ such that
\[
\frac{
\ell_{\rec}(n_j;\varepsilon,\delta)
}{
n_j^{2/3}
}
\longrightarrow0.
\]
The universal subcritical theorem implies that, along this sequence, the optimal fidelity remains asymptotically at the baseline $4^{-k}$, contradicting the definition of
$\ell_{\rec}$.

For the upper bound, suppose that no finite constant in
\eqref{eq:ordered-theta} existed. Choose $C_j\to\infty$ and a sequence $n_j\to\infty$ such that
\[
C_j n_j^{-1/3}\longrightarrow0
\]
and
\[
\ell_{\rec}(n_j;\varepsilon,\delta)
>
C_j n_j^{2/3}.
\]
Set
\[
\ell_j
=
\left\lfloor
C_j n_j^{2/3}
\right\rfloor.
\]
Then
\[
\frac{\ell_j}{n_j^{2/3}}\longrightarrow\infty,
\qquad
\frac{\ell_j}{n_j}\longrightarrow0.
\]
The supercritical theorem therefore implies that
\[
\Prob_{U\sim\mathrm{Haar}}
\left[
\Fopt(U;\ell_j)\geq1-\varepsilon
\right]
\longrightarrow1.
\]
For all sufficiently large $j$, this probability exceeds $1-\delta$, contradicting
$\ell_j<\ell_{\rec}(n_j;\varepsilon,\delta)$.

The two bounds prove Corollary~\ref{cor:ordered-operational}. No information about the limiting crossover at
$\ell=c\,n^{2/3}$ is needed.

\section{Full proof of the block-resolved phase diagram}
\label{app:block-full}

This appendix proves Theorem~\ref{thm:block} and derives the phase diagram in Corollary~\ref{cor:phase}.

\subsection{Model and operational quantities}

Let $n=Bm$, with the microscopic positions divided into $B$ consecutive blocks
$I_1,\ldots,I_B$, each of size $m$. A uniformly random subset
$S\subset[n]$ of size $\ell$ survives. Define
\[
L_b=|S\cap I_b|,
\qquad
L=(L_1,\ldots,L_B).
\]
The receiver is given the count vector $L$ in a classical register $\mathsf L$. The surviving qubits are identified with a common output register and arranged block by block, in increasing microscopic order within each block. The channel is
\[
\D^{(B)}_{n\to\ell}(\rho)
=
\frac1{\binom n\ell}
\sum_{\substack{S\subset[n]\\|S|=\ell}}
|L(S)\rangle\!\langle L(S)|_{\mathsf L}
\otimes
J_S\Tr_{S^c}(\rho)J_S^\dagger,
\]
where $J_S$ is the canonical order-preserving identification.

A feasible count vector satisfies
\[
0\leq L_b\leq m,
\qquad
\sum_{b=1}^B L_b=\ell,
\]
and occurs with probability
\[
p_L
=
\frac{\displaystyle\prod_{b=1}^B\binom m{L_b}}
{\displaystyle\binom n\ell}.
\]
Conditioned on $L$, the surviving subsets in different blocks are independent and uniformly distributed. Up to the canonical identifications of the output registers,
\begin{equation}
\left.\D^{(B)}_{n\to\ell}\right|_L
\simeq
\bigotimes_{b=1}^B\D_{m\to L_b}.
\label{eq:factorization}
\end{equation}

The input state is
\[
\rho_{CR}
=
\Phi_{MR}\otimes\pi_H,
\qquad
C=MH,
\]
and $E$ purifies $H$, as in \eqref{eq:hp-input}. After applying the known scrambling unitary $U$ and the flagged channel, the accessible state is
$\rho^U_{RDE\mathsf L}$. The optimal fidelity is given by
\eqref{eq:fopt-def}, with a recovery channel
\[
\mathcal R:DE\mathsf L\longrightarrow\widehat M
\]
that may depend on $U,n,\ell$, and $B$, and may be controlled by the observed value of $L$.

For a single block containing $m$ microscopic positions and $r$ survivors, define the order-statistic kernel
\[
p^{(m,r)}_{ai}
=
\frac{\binom{i-1}{a-1}\binom{m-i}{r-a}}
{\binom mr},
\qquad
1\leq a\leq r,
\]
and its collision strength
\[
V_{m,r}
=
\sum_{a=1}^r\sum_{i=1}^m
\left(p^{(m,r)}_{ai}\right)^2,
\qquad
V_{m,0}=0.
\]
For a branch $L$, let
\[
X(L)
=
\sum_{b=1}^B V_{m,L_b}.
\]

\subsection{Collision bounds for discrete log-concave laws}

\begin{lemma}[Log-concave anti-concentration]
Let $(\pi_j)_{j\in\mathbb Z}$ be a log-concave probability mass function with interval support. Let
\[
\sigma^2=\operatorname{Var}(Z),
\qquad
M=\max_j\pi_j.
\]
There are universal constants $c,C>0$ such that
\begin{equation}
\frac{c}{1+\sigma}
\leq
\sum_j\pi_j^2
\leq
M
\leq
\frac{C}{1+\sigma}.
\label{eq:logconcave-collision}
\end{equation}
In particular,
\begin{equation}
M\leq C\sum_j\pi_j^2.
\label{eq:logconcave-max-collision}
\end{equation}
\end{lemma}

\begin{proof}
The middle inequality follows from
\[
\sum_j\pi_j^2
\leq
M\sum_j\pi_j
=
M.
\]
If $\sigma=0$, the distribution is concentrated at one point and the result is immediate. Suppose that $\sigma>0$. Chebyshev's inequality places at least $3/4$ of the mass in an interval containing at most $4\sigma+3$ integers. Cauchy--Schwarz therefore gives
\[
\sum_j\pi_j^2
\geq
\frac{(3/4)^2}{4\sigma+3}
\geq
\frac{c}{1+\sigma}.
\]

It remains to bound $M$ from above. Translate a mode to the origin, so that
$\pi_0=M$. On the nonnegative half-line, let $t$ be the first positive integer for which
$\pi_t\leq M/2$. If the support ends before this happens, take $t$ to be the first point outside the support. Since
\[
\pi_j>M/2
\qquad
(0\leq j<t),
\]
normalization implies
\[
t\leq\frac{2}{M}.
\]

If $t$ lies outside the support, the nonnegative tail vanishes and there is nothing to prove. Otherwise, set
\[
r_j=\frac{\pi_{j+1}}{\pi_j}
\qquad(0\leq j<t).
\]
Log-concavity implies that $(r_j)$ is nonincreasing. Since
\[
\prod_{j=0}^{t-1}r_j
=
\frac{\pi_t}{\pi_0}
\leq
\frac12,
\]
we have $r_{t-1}\leq2^{-1/t}$. Consequently, for every $s\geq0$,
\[
\pi_{t+s}
\leq
\pi_t\,2^{-s/t}
\leq
\frac{M}{2}\,2^{-s/t}.
\]
The contribution of $0\leq j<t$ to $\sum_j j^2\pi_j$ is at most $t^2$, while summing the geometric tail gives a contribution at most $CMt^3$. Both terms are $O(M^{-2})$. Applying the same argument to the negative half-line gives
\[
\sigma^2
\leq
\E Z^2
\leq
CM^{-2}.
\]
Thus $M\leq C/\sigma$ when $\sigma\geq1$. For $\sigma<1$, the trivial bound
$M\leq1$ completes \eqref{eq:logconcave-collision}.
\end{proof}

\subsection{Uniform single-block estimates}

For fixed rank $a$, the shifted order statistic
\[
Z_a=S_a-a
\]
has the beta-binomial mass
\[
\pi_a(z)
=
\frac{
\binom{z+a-1}{a-1}
\binom{m-a-z}{r-a}
}{
\binom mr
},
\qquad
0\leq z\leq m-r.
\]
The ratio of consecutive masses is a product of two nonincreasing factors, so
$\pi_a$ is log-concave. Its variance is
\begin{equation}
\sigma_a^2
=
\frac{(m-r)a(r+1-a)(m+1)}
{(r+1)^2(r+2)}.
\label{eq:beta-binomial-variance}
\end{equation}

\begin{lemma}[Uniform collision scale]
There are absolute constants $c,C>0$ such that, for every
$1\leq r\leq m$,
\begin{equation}
c\,
\frac{r^{3/2}}
{\sqrt{m(m-r+1)}}
\leq
V_{m,r}
\leq
C\,
\frac{r^{3/2}}
{\sqrt{m(m-r+1)}}.
\label{eq:uniform-V}
\end{equation}
In particular, when $r=o(m)$,
\begin{equation}
V_{m,r}
\asymp
\frac{r^{3/2}}{m}.
\label{eq:sparse-block-V}
\end{equation}
\end{lemma}

\begin{proof}
Applying the preceding lemma to each order-statistic row gives
\[
\sum_i\left(p_{ai}^{(m,r)}\right)^2
\asymp
\frac1{1+\sigma_a},
\]
with constants independent of $m,r$, and $a$.

Set
\[
d=m-r,
\qquad
R=r+1.
\]
If $d=0$, every order statistic is deterministic. Hence
\[
V_{m,r}=r,
\]
which agrees with \eqref{eq:uniform-V} up to universal constants.

Suppose that $d\geq1$. From
\eqref{eq:beta-binomial-variance},
\[
\sigma_a
=
A\sqrt{a(R-a)},
\qquad
A
=
\frac{\sqrt{d(m+1)}}{R\sqrt{r+2}}.
\]
The case $r=1$ follows directly from
\[
V_{m,1}
=
\sum_{i=1}^m\frac1{m^2}
=
\frac1m.
\]
We may therefore assume that $r\geq2$.

By symmetry, and using
\[
\frac{aR}{2}
\leq
a(R-a)
\leq
aR
\qquad
\left(1\leq a\leq\frac R2\right),
\]
we obtain
\[
\sum_{a=1}^r\frac1{1+\sigma_a}
\asymp
\sum_{a\leq r/2}
\frac1{1+\beta\sqrt a},
\qquad
\beta=A\sqrt R.
\]
An integral comparison, or a decomposition at
$a=\beta^{-2}$, gives
\[
\sum_{a\leq r/2}
\frac1{1+\beta\sqrt a}
\asymp
\frac{r}{1+\beta\sqrt r}.
\]
Moreover,
\[
(\beta\sqrt r)^2
=
\frac{d(m+1)r}{(r+1)(r+2)}
\asymp
\frac{dm}{r}.
\]
For $d\geq1$, this quantity is bounded below by a positive universal constant. It follows that
\[
\frac{r}{1+\beta\sqrt r}
\asymp
\frac{r^{3/2}}{\sqrt{dm}}
\asymp
\frac{r^{3/2}}{\sqrt{m(d+1)}}.
\]
Since $d+1=m-r+1$, summing the row collision probabilities proves
\eqref{eq:uniform-V}.
\end{proof}

\begin{lemma}[Pointwise score control]
For a uniformly distributed ordered $r$-subset
$S=(S_1<\cdots<S_r)$ of $[m]$, define
\[
q_{m,r}(S)
=
\sum_{a=1}^r p^{(m,r)}_{a,S_a}.
\]
Then
\begin{equation}
\E_S q_{m,r}(S)
=
V_{m,r},
\qquad
0\leq q_{m,r}(S)\leq C V_{m,r}
\label{eq:pointwise-score-control}
\end{equation}
for every $S$.
\end{lemma}

\begin{proof}
Expanding the expectation gives
\[
\E_S q_{m,r}(S)
=
\sum_{a=1}^r\sum_{i=1}^m
p^{(m,r)}_{ai}
\Prb(S_a=i)
=
\sum_{a=1}^r\sum_{i=1}^m
\left(p^{(m,r)}_{ai}\right)^2
=
V_{m,r}.
\]

For each rank $a$, the log-concave collision bound gives
\[
\max_i p^{(m,r)}_{ai}
\leq
C\sum_i\left(p^{(m,r)}_{ai}\right)^2.
\]
Therefore
\[
q_{m,r}(S)
\leq
\sum_{a=1}^r\max_i p^{(m,r)}_{ai}
\leq
C V_{m,r},
\]
which proves the pointwise estimate.
\end{proof}

\subsection{The scale set by coarse side information}

Set
\[
u=\frac{\ell}{B},
\qquad
\Lambda_{n,\ell,B}
=
\begin{cases}
B\ell/n, & u\leq1,\\[2mm]
\sqrt B\,\ell^{3/2}/n, & u\geq1.
\end{cases}
\]

\begin{lemma}[Deterministic lower bound and averaged upper bound]
Assume that $\ell=o(n)$. For every feasible count vector $L$,
\begin{equation}
X(L)\geq c\Lambda_{n,\ell,B}.
\label{eq:det-lower}
\end{equation}
If $L$ has the multivariate hypergeometric law induced by a uniformly chosen surviving set, then
\begin{equation}
\E_LX(L)\leq C\Lambda_{n,\ell,B}.
\label{eq:mean-upper}
\end{equation}
\end{lemma}

\begin{proof}
Each order-statistic row is supported on at most $m$ points, so its collision probability is at least $1/m$. It follows that
\[
V_{m,r}\geq\frac{r}{m}.
\]
When $u\leq1$,
\[
X(L)
\geq
\frac1m\sum_{b=1}^B L_b
=
\frac{B\ell}{n}.
\]

When $u\geq1$, the lower bound in \eqref{eq:uniform-V} gives
\[
V_{m,r}
\geq
c\frac{r^{3/2}}{m}.
\]
The convexity of $r^{3/2}$ then yields
\[
X(L)
\geq
\frac{c}{m}\sum_{b=1}^B L_b^{3/2}
\geq
\frac{c}{m}
B\left(\frac{\ell}{B}\right)^{3/2}
=
c\frac{\sqrt B\,\ell^{3/2}}{n}.
\]
This proves \eqref{eq:det-lower}.

For the upper bound, let
\[
H\sim\operatorname{Hyp}(n,m,\ell).
\]
By symmetry of the blocks,
\[
\E_LX(L)
=
B\,\E V_{m,H},
\qquad
\E H=u.
\]
The identity
\[
p_{ai}^{(m,r)}
=
\frac{r}{m}
\Prb\!\left\{
\operatorname{Hyp}(m-1,i-1,r-1)=a-1
\right\}
\]
implies
\[
\max_i p_{ai}^{(m,r)}
\leq
\frac{r}{m}.
\]
Since the entries in each row sum to one,
\[
V_{m,r}
=
\sum_a\sum_i\left(p_{ai}^{(m,r)}\right)^2
\leq
\sum_a\max_i p_{ai}^{(m,r)}
\leq
\frac{r^2}{m}.
\]

Suppose first that $u\leq1$. Since
\[
\operatorname{Var}(H)\leq u,
\]
we have
\[
\E H^2
\leq
u^2+u
\leq
2u.
\]
Therefore
\[
\E V_{m,H}
\leq
\frac{\E H^2}{m}
\leq
C\frac{u}{m},
\]
and multiplication by $B$ gives the first branch of
\eqref{eq:mean-upper}.

Now assume that $u\geq1$. On the event $H\leq m/2$, the upper bound in
\eqref{eq:uniform-V} gives
\[
V_{m,H}
\leq
C\frac{H^{3/2}}{m}.
\]
Moreover,
\[
\E H^2
\leq
u^2+u
\leq
2u^2.
\]
Lyapunov's inequality therefore gives
\[
\E\!\left[
H^{3/2}\bm1_{\{H\leq m/2\}}
\right]
\leq
\E H^{3/2}
\leq
(\E H^2)^{3/4}
\leq
C u^{3/2}.
\]

It remains to control the event $H>m/2$. Since
\[
p=\frac{\ell}{n}=\frac{u}{m}=o(1),
\]
a hypergeometric Chernoff bound gives
\[
\Prb(H>m/2)
\leq
(2ep)^{m/2}
\]
for all sufficiently large $n$. On this event we use the trivial bound
\[
V_{m,H}\leq H\leq m.
\]
Because $u=mp\geq1$, one has $m\geq p^{-1}$. Furthermore,
\[
\frac{
m(2ep)^{m/2}
}{
u^{3/2}/m
}
=
\frac{\sqrt m}{p^{3/2}}(2ep)^{m/2}
\longrightarrow0.
\]
Indeed, the logarithm of the right-hand side is
\[
\frac m2\log(2ep)
+\frac12\log m
-\frac32\log p.
\]
Since $m\geq p^{-1}$, the magnitude of the negative first term is at least of order
$|\log p|/p$, which dominates the remaining logarithmic terms as $p\to0$. Hence
\[
m\,\Prb(H>m/2)
=
o\!\left(\frac{u^{3/2}}{m}\right).
\]
Combining the two parts gives
\[
\E V_{m,H}
\leq
C\frac{u^{3/2}}{m}.
\]
Multiplying by $B$ proves the second branch of
\eqref{eq:mean-upper}.
\end{proof}

\subsection{Branchwise converse}

Let $\tau_{m,r}$ be the normalized Choi state of
$\D_{m\to r}$ and define
\[
G_{m,r}
=
2^{m+r}\Tr\!\left(\tau_{m,r}^2\right).
\]
For $r=0$, one has $G_{m,0}=1$. For $r\geq1$, the uniform factorial-moment estimate
\eqref{eq:factorial}, together with \eqref{eq:uniform-V}, implies that there are constants
$v_0,C>0$ such that
\begin{equation}
V_{m,r}\leq v_0
\quad\Longrightarrow\quad
G_{m,r}-1
\leq
C V_{m,r}.
\label{eq:G-V}
\end{equation}

To verify that the constants are uniform, choose $v_0$ sufficiently small. Since
\[
V_{m,r}\geq\frac{r}{m},
\]
the condition $V_{m,r}\leq v_0$ then implies $r/m\leq1/2$. In this range,
\eqref{eq:uniform-V} gives
\[
\frac{(r+1)^{3/2}}{m}
\leq
C_0V_{m,r}.
\]
Substituting the factorial-moment estimate into
\[
4^M
=
\sum_{s\geq0}
3^s\frac{(M)_s}{s!}
\]
and summing the resulting geometric series proves \eqref{eq:G-V}.

\begin{proposition}[Branchwise converse]
There are constants $v_0,C>0$ such that every feasible branch satisfying
$X(L)\leq v_0$ obeys
\begin{equation}
F_{\rm e}^{\rm opt}(U\mid L)
\leq
4^{-k}
+
\frac12
\sqrt{
e^{CX(L)}-1
}
\label{eq:branch-converse-bound}
\end{equation}
for every scrambling unitary $U$. Consequently, for any sequence of feasible branches with
$X(L)\to0$,
\[
F_{\rm e}^{\rm opt}(U\mid L)
\longrightarrow
4^{-k}
\]
uniformly over $U$.
\end{proposition}

\begin{proof}
Conditioned on $L$, the channel factorizes according to
\eqref{eq:factorization}. Its normalized Choi state is therefore a tensor product, and its normalized purity is
\[
G(L)
=
\prod_{b=1}^B G_{m,L_b}.
\]
If $X(L)\leq v_0$, then
\[
V_{m,L_b}
\leq
X(L)
\leq
v_0
\]
for every block. Equation~\eqref{eq:G-V} gives
\[
G(L)
\leq
\prod_{b=1}^B
\left(1+CV_{m,L_b}\right)
\leq
\exp\left(
C\sum_{b=1}^B V_{m,L_b}
\right)
=
e^{CX(L)}.
\]
The Choi trace-distance converse from
Proposition~\ref{prop:converse} then gives
\[
F_{\rm e}^{\rm opt}(U\mid L)
\leq
4^{-k}
+
\frac12\sqrt{G(L)-1},
\]
which implies \eqref{eq:branch-converse-bound}. The lower bound
\[
F_{\rm e}^{\rm opt}(U\mid L)\geq4^{-k}
\]
is achieved by discarding the decoder input and preparing the maximally mixed diary. The stated limit follows.
\end{proof}

\begin{proposition}[Global converse]
If
\[
\Lambda_{n,\ell,B}\longrightarrow0,
\]
then
\[
F_{\rm e}^{\rm opt}(U;\ell,B)
\longrightarrow
4^{-k}
\]
uniformly over the scrambling unitary $U$.
\end{proposition}

\begin{proof}
Equation~\eqref{eq:mean-upper} gives
\[
\E_LX(L)\longrightarrow0.
\]
Choose a deterministic sequence $\eta_n\downarrow0$ such that
\[
\Prb\{X(L)>\eta_n\}\longrightarrow0.
\]
For example, whenever $\E_LX(L)>0$, one may take
\[
\eta_n=\sqrt{\E_LX(L)}
\]
and apply Markov's inequality.

Because the flag is classical, the recovery channel may be chosen independently on each orthogonal branch. Hence
\[
F_{\rm e}^{\rm opt}(U;\ell,B)
=
\sum_L
p_L
F_{\rm e}^{\rm opt}(U\mid L).
\]
For sufficiently large $n$, one has $\eta_n\leq v_0$. On the event
$X(L)\leq\eta_n$, Equation~\ref{eq:branch-converse-bound} gives
\[
F_{\rm e}^{\rm opt}(U\mid L)
\leq
4^{-k}
+
\frac12
\sqrt{
e^{C\eta_n}-1
}.
\]
On the complementary event, we use the bound
$F_{\rm e}^{\rm opt}\leq1$. Therefore
\begin{align}
F_{\rm e}^{\rm opt}(U;\ell,B)
&\leq
4^{-k}
+
\frac12\sqrt{e^{C\eta_n}-1}
+
\Prb\{X(L)>\eta_n\}
\\
&=
4^{-k}+o(1),
\end{align}
uniformly over $U$. The trivial decoder gives the matching lower bound
$4^{-k}$, completing the proof.
\end{proof}

\subsection{Branchwise achievability}

Fix a feasible branch $L$. For each block, let
\[
P_b=\left(p_{ai}^{(m,L_b)}\right)_{a,i},
\qquad
V_b=V_{m,L_b},
\]
and define
\[
V=\sum_{b=1}^B V_b=X(L),
\qquad
V_{\max}=\max_{1\leq b\leq B}V_b.
\]
After computational-basis measurement, the conditional channel is a product of independent planted-subsequence models. We denote the input and output Rademacher vectors in block $b$ by $\xi_b$ and $\eta_b$.

\begin{lemma}[Reduction to a single block]
\label{lem:single-block-reduction}
Let $(m,r)$ be any sequence such that
\[
V_{m,r}\longrightarrow\infty.
\]
Then the measured Choi distribution of $\D_{m\to r}$ satisfies
\[
I_{1/2}^{\rm S}(X:Y)\longrightarrow\infty.
\]
\end{lemma}

\begin{proof}
The assumption implies $m\to\infty$. Suppose first that
$r\geq m^{3/4}$. Uniformly delete output systems until
\[
r_0=\lfloor m^{3/4}\rfloor
\]
survivors remain. Then
\[
r_0=o(m),
\qquad
\frac{r_0}{m^{2/3}}\longrightarrow\infty.
\]
Proposition~\ref{prop:sibson} applies to the resulting ordered-deletion channel. Since the channel with $r_0$ survivors is a degradation of the one with $r$ survivors, data processing gives the result for the original output.

Now suppose that $r<m^{3/4}$. In this case $r=o(m)$, and
\eqref{eq:uniform-V} gives
\[
V_{m,r}
\leq
C\frac{r^{3/2}}{m}.
\]
Thus $V_{m,r}\to\infty$ implies
\[
\frac{r}{m^{2/3}}\longrightarrow\infty,
\]
so Proposition~\ref{prop:sibson} applies directly.
\end{proof}

For a fixed branch $L$, let $P_{XY}^{(L)}$ denote the measured planted distribution and let $\tau^{(L)}_{C'D}$ be the corresponding conditional Choi state.

\begin{proposition}[Information divergence on product branches]
\label{prop:branch-information}
For every deterministic sequence of feasible branches $L=L_n$ such that
\[
X(L)\longrightarrow\infty,
\]
one has
\begin{equation}
I_{1/2}^{\rm S}(X:Y)_{P^{(L)}}
\longrightarrow\infty.
\label{eq:branch-sibson-divergence}
\end{equation}
Consequently,
\begin{equation}
I_{2/3}^{*}(C':D)_{\tau^{(L)}}
\longrightarrow\infty.
\label{eq:branch-quantum-divergence}
\end{equation}
\end{proposition}

\begin{proof}
If
\[
V_{\max}>\sqrt V,
\]
then some block satisfies
\[
V_b>\sqrt V\longrightarrow\infty.
\]
Keeping only that block and applying
Lemma~\ref{lem:single-block-reduction} proves
\eqref{eq:branch-sibson-divergence} by data processing.

It remains to consider the case
\[
V_{\max}\leq\sqrt V.
\]
Define the global alignment score
\[
T
=
\sum_{b=1}^B
\xi_b^T P_b^T\eta_b.
\]
Conditioned on the hidden subsets $S_b$, one has
\[
\eta_b=R_{S_b}\xi_b,
\]
where $R_{S_b}$ is the corresponding selection matrix. Therefore
\[
\E[T\mid S,L]
=
Q(S,L)
:=
\sum_{b=1}^B q_{m,L_b}(S_b).
\]

The subsets $S_b$ are independent conditioned on $L$, and
\[
\E[Q\mid L]=V.
\]
By the pointwise score bound,
\[
0\leq q_{m,L_b}(S_b)\leq C V_b,
\qquad
\E q_{m,L_b}(S_b)=V_b.
\]
Hence
\begin{align}
\operatorname{Var}(Q\mid L)
&=
\sum_b
\operatorname{Var}\!\left(q_{m,L_b}(S_b)\right)
\\
&\leq
C\sum_bV_b^2
\\
&\leq
C V_{\max}V
\\
&\leq
C V^{3/2}
=
o(V^2).
\end{align}
It follows that
\begin{equation}
\Prb\!\left(Q\geq\frac V2\,\middle|\,L\right)
\longrightarrow1.
\label{eq:Q-concentration}
\end{equation}

For each block, define
\[
A_b
=
\frac12
\left(
P_b^TR_{S_b}+R_{S_b}^TP_b
\right),
\qquad
A=\bigoplus_{b=1}^B A_b.
\]
Then
\[
T=\xi^TA\xi,
\qquad
\operatorname{tr}A=Q(S,L).
\]
Moreover,
\begin{align}
\norm{A}_F^2
&\leq
\sum_b
\norm{P_b^TR_{S_b}}_F^2
=
\sum_b\norm{P_b}_F^2
=
V,
\\
\norm{A}_2
&\leq
\max_b\norm{P_b}_2
=
\max_b\sqrt{\frac{L_b}{m}}
\leq1.
\end{align}
Indeed, each $P_b$ has row sums equal to one and column sums equal to
$L_b/m$. Hence
\[
\norm{P_b}_2
\leq
\sqrt{\norm{P_b}_1\norm{P_b}_\infty}
=
\sqrt{\frac{L_b}{m}},
\]
and the normalized constant vectors attain equality.
On the event $Q\geq V/2$, the Hanson--Wright inequality gives
\[
\Prb\!\left(
T<\frac V4
\,\middle|\,
S,L
\right)
\leq
2\exp\left[
-c\min\left\{
\frac{V^2}{V},
\frac{V}{1}
\right\}
\right]
\leq
2e^{-cV}.
\]
Together with \eqref{eq:Q-concentration}, this shows that
\begin{equation}
P^{(L)}
\left(
T\geq\frac V4
\right)
\longrightarrow1.
\label{eq:global-score-large}
\end{equation}

For an output word $y$, define
\[
v(y)
=
\sum_{b=1}^B
\norm{P_b^T\eta_b(y)}_2^2.
\]
The planted output marginal is uniform, and therefore
\[
\E_{P^{(L)}}v(Y)
=
\sum_b\norm{P_b}_F^2
=
V.
\]
Consider the event
\[
\mathcal A_L
=
\left\{
T\geq\frac V4
\right\}
\cap
\left\{
v(Y)\leq V^{3/2}
\right\}.
\]
Equation~\eqref{eq:global-score-large} and Markov's inequality give
\[
P^{(L)}(\mathcal A_L)
\longrightarrow1.
\]

Under $P_XQ_Y$, where $P_X$ is the fixed uniform input marginal and $Q_Y$ is arbitrary, conditionally on $Y=y$,
\[
T
=
\sum_b
\xi_b^TP_b^T\eta_b(y)
\]
is a centered Rademacher sum with variance proxy $v(y)$. On
$\mathcal A_L$, Hoeffding's inequality gives
\[
(P_XQ_Y)(\mathcal A_L)
\leq
\exp\left(
-c\frac{V^2}{V^{3/2}}
\right)
=
e^{-c\sqrt V},
\]
uniformly over $Q_Y$. Thus
\[
\sup_{Q_Y}
(P_XQ_Y)(\mathcal A_L)
\longrightarrow0.
\]

The standard event decomposition for the Hellinger affinity now gives
\[
\sup_{Q_Y}
\operatorname{Aff}
\left(
P_{XY}^{(L)},
P_XQ_Y
\right)
\longrightarrow0.
\]
Since
\[
D_{1/2}(P\|R)
=
-2\log_2\operatorname{Aff}(P,R),
\]
this proves \eqref{eq:branch-sibson-divergence}. Monotonicity in the R\'enyi order and computational-basis data processing then give
\eqref{eq:branch-quantum-divergence}.
\end{proof}

\subsection{Proof of the block-resolved threshold}

\begin{proof}[Proof of Theorem~\ref{thm:block}]
The subcritical statement follows from the global converse proved above.

Suppose now that
\[
\Lambda_{n,\ell,B}\longrightarrow\infty.
\]
The deterministic lower bound \eqref{eq:det-lower} gives
\[
\inf_{\substack{L\ \mathrm{feasible}\\\sum_bL_b=\ell}}
X(L)
\longrightarrow\infty.
\]

For a feasible branch, define
\[
\alpha_n(L)
=
\sup_{Q_Y}
\operatorname{Aff}
\left(
P_{XY}^{(L)},
P_XQ_Y
\right).
\]
Proposition~\ref{prop:branch-information} implies
\begin{equation}
\sup_L\alpha_n(L)\longrightarrow0.
\label{eq:uniform-branch-affinity}
\end{equation}
Indeed, otherwise there would be a subsequence and feasible branches
$L_n$ for which $\alpha_n(L_n)$ remained bounded away from zero. The deterministic lower bound would give $X(L_n)\to\infty$, contradicting the proposition.

Let $P_{XY\mathsf L}$ be the full measured distribution, including the accessible classical flag. For an arbitrary distribution
\[
Q_{Y\mathsf L}
=
\sum_L q_L Q_{Y|L}\otimes|L\rangle\!\langle L|,
\]
the affinity decomposes as
\begin{align}
\operatorname{Aff}
\left(
P_{XY\mathsf L},
P_XQ_{Y\mathsf L}
\right)
&=
\sum_L
\sqrt{p_Lq_L}\,
\operatorname{Aff}
\left(
P_{XY}^{(L)},
P_XQ_{Y|L}
\right)
\\
&\leq
\left(\sup_L\alpha_n(L)\right)
\sum_L\sqrt{p_Lq_L}
\\
&\leq
\sup_L\alpha_n(L).
\end{align}
Here the last inequality follows from Cauchy--Schwarz. Equation~\eqref{eq:uniform-branch-affinity} therefore implies
\[
I_{1/2}^{\rm S}(X:Y\mathsf L)
\longrightarrow\infty.
\]
Monotonicity in the R\'enyi order and data processing under computational-basis measurement give
\begin{equation}
I_{2/3}^{*}(C':D\mathsf L)_\tau
\longrightarrow\infty.
\label{eq:flagged-quantum-information}
\end{equation}

Let $F_{\rm env}$ denote the environment of a pure Choi dilation of the flagged channel
\[
C'\longrightarrow D\mathsf L.
\]
Pure-state duality gives
\[
H_2^*(C'|F_{\rm env})
=
-H_{2/3}^*(C'|D\mathsf L).
\]
Since the input Choi marginal is maximally mixed,
\[
H_{2/3}^*(C'|D\mathsf L)
=
n-I_{2/3}^*(C':D\mathsf L)_\tau.
\]
Together with
\[
H_2^*(C|R)=n-2k,
\]
this gives
\[
H_2^*(C|R)
+
H_2^*(C'|F_{\rm env})
=
I_{2/3}^*(C':D\mathsf L)_\tau-2k.
\]

The joint state--channel decoupling theorem of
Ref.~\cite{ChengDupuisGao2024}, at R\'enyi order two, bounds the mean complementary-channel error by
\[
2^{-\frac12[
I_{2/3}^*(C':D\mathsf L)_\tau
-2k+\log_2 3]}.
\]
By \eqref{eq:flagged-quantum-information}, this quantity tends to zero. Markov's inequality gives decoupling in probability over the Haar-random scrambling unitary. After restoring the early-radiation purifier $E$, Uhlmann's theorem provides a recovery channel on
$DE\mathsf L$ whose entanglement fidelity tends to one.
\end{proof}

\subsection{Operational recovery scale}

The block-resolved channels are degraded as the number of survivors decreases. More precisely, uniformly deleting one of the $\ell+1$ output systems and decrementing the corresponding block count defines a channel
$\mathcal G_{\ell+1\to\ell}^{(B)}$ such that
\begin{equation}
\D^{(B)}_{n\to\ell}
=
\mathcal G_{\ell+1\to\ell}^{(B)}
\circ
\D^{(B)}_{n\to\ell+1}.
\label{eq:block-degradation-full}
\end{equation}
Hence
\[
F_{\rm e}^{\rm opt}(U;\ell,B)
\]
is nondecreasing in $\ell$.

\begin{proof}[Proof of Corollary~\ref{cor:phase}]
Recall
\[
s(n,B)
=
\begin{cases}
n^{2/3}B^{-1/3}, & B\leq\sqrt n,\\[1mm]
n/B, & B\geq\sqrt n.
\end{cases}
\]
At $\ell=s(n,B)$, the corresponding value of
$\Lambda_{n,\ell,B}$ is of order one. More precisely, the first branch of $s$ lies in the regime $\ell/B\geq1$, while the second lies in the regime $\ell/B\leq1$. Since
$\Lambda_{n,\ell,B}$ is continuous and increasing in $\ell$,
\begin{align}
\frac{\ell}{s(n,B)}\longrightarrow0
&\quad\Longrightarrow\quad
\Lambda_{n,\ell,B}\longrightarrow0,
\label{eq:scale-to-subcritical}
\\
\frac{\ell}{s(n,B)}\longrightarrow\infty,
\qquad
\ell=o(n)
&\quad\Longrightarrow\quad
\Lambda_{n,\ell,B}\longrightarrow\infty.
\label{eq:scale-to-supercritical}
\end{align}

Suppose that no uniform positive lower constant existed. Then there would be a sequence
$c_j\downarrow0$ and pairs $(n_j,B_j)$, with $B_j\mid n_j$, such that
\[
\ell_{\rec}(n_j,B_j;\varepsilon,\delta)
<
c_j s(n_j,B_j).
\]
Since $\ell_{\rec}\geq1$, this forces
\[
s(n_j,B_j)\longrightarrow\infty.
\]
Consequently,
\[
\frac{
\ell_{\rec}(n_j,B_j;\varepsilon,\delta)
}{
s(n_j,B_j)
}
\longrightarrow0.
\]
Because $s(n,B)\leq n^{2/3}$, this sequence also satisfies
$\ell_{\rec}=o(n)$. Equation~\eqref{eq:scale-to-subcritical} and the subcritical part of Theorem~\ref{thm:block} imply that the optimal fidelity approaches $4^{-k}$ uniformly over $U$. This contradicts the defining fixed-error property of
$\ell_{\rec}$, since
\[
1-\varepsilon>4^{-k}.
\]

For the upper bound, suppose that no uniform finite constant existed. Choose
\[
C_j\longrightarrow\infty
\]
and $N_j\to\infty$ such that
\[
C_jN_j^{-1/3}\longrightarrow0.
\]
The assumed failure of a uniform bound allows us to choose
$n_j\geq N_j$ and $B_j\mid n_j$ such that
\[
\ell_{\rec}(n_j,B_j;\varepsilon,\delta)
>
C_j s(n_j,B_j).
\]
Set
\[
\ell_j
=
\left\lfloor
C_j s(n_j,B_j)
\right\rfloor.
\]
Since $s(n,B)\leq n^{2/3}$,
\[
\frac{\ell_j}{n_j}
\leq
C_jn_j^{-1/3}
\leq
C_jN_j^{-1/3}
\longrightarrow0.
\]
Moreover,
\[
\frac{\ell_j}{s(n_j,B_j)}
\longrightarrow\infty.
\]
Equation~\eqref{eq:scale-to-supercritical} gives
\[
\Lambda_{n_j,\ell_j,B_j}\longrightarrow\infty.
\]
The supercritical part of Theorem~\ref{thm:block} therefore implies
\[
\Prb_{U\sim\mathrm{Haar}}
\left[
F_{\rm e}^{\rm opt}(U;\ell_j,B_j)
\geq1-\varepsilon
\right]
\longrightarrow1.
\]
For all sufficiently large $j$, this probability exceeds $1-\delta$, contradicting
\[
\ell_j
<
\ell_{\rec}(n_j,B_j;\varepsilon,\delta).
\]

Thus there are constants
\[
0<c_{k,\varepsilon,\delta}
\leq
C_{k,\varepsilon,\delta}
<\infty
\]
such that
\[
c_{k,\varepsilon,\delta}s(n,B)
\leq
\ell_{\rec}(n,B;\varepsilon,\delta)
\leq
C_{k,\varepsilon,\delta}s(n,B)
\]
for all sufficiently large $n$, uniformly over divisors $B$ of $n$.
\end{proof}

\section{Exact spectral decomposition of the chronological channel}
\label{app:chron-exact}

This appendix derives the operator decomposition used in
Sec.~\ref{sec:chronological}. For
\[
\bm m=(m_x,m_y,m_z),
\qquad
|\bm m|=m_x+m_y+m_z=r,
\]
let $\mathcal P_{\bm m}^{(n)}$ be the set of $n$-qubit Pauli strings containing $m_x$ copies of $X$, $m_y$ copies of $Y$, $m_z$ copies of $Z$, and identities on the remaining sites. Define
\begin{equation}
 F_{\bm m}^{(n)}
 =
 \left[
 \frac{n!}{m_x!m_y!m_z!(n-r)!}
 \right]^{-1/2}
 \sum_{P\in\mathcal P_{\bm m}^{(n)}} P.
 \label{eq:symmetric-pauli-normalization}
\end{equation}
These operators are orthonormal with respect to the normalized Hilbert--Schmidt inner product,
\begin{equation}
 2^{-n}\Tr\left[
 (F_{\bm m}^{(n)})^\dagger F_{\bm m'}^{(n)}
 \right]
 =\delta_{\bm m,\bm m'}.
\end{equation}

The permutation twirl projects onto the permutation-invariant operator subspace spanned by these symmetric Pauli sums. Counting the terms that survive the partial trace gives
\begin{equation}
 \Ochan_{n\to\ell}\!\left(2^{-n}F_{\bm m}^{(n)}\right)
 =
 2^{-\ell}
 \sqrt{\frac{(\ell)_r}{(n)_r}}\,
 F_{\bm m}^{(\ell)}.
 \label{eq:chron-action-app}
\end{equation}
For each degree $r$, the number of admissible multiindices is
\begin{equation}
 d_r
 =
 \#\left\{
 (m_x,m_y,m_z)\in\mathbb Z_{\geq0}^3:
 m_x+m_y+m_z=r
 \right\}
 =
 \binom{r+2}{2}.
\end{equation}
Let $\tau_{C'D}^{\rm chron}$ be the normalized Choi state of $\Ochan_{n\to\ell}$ and set
\begin{equation}
 L_{n,\ell}^{\rm chron}
 :=2^{n+\ell}\tau_{C'D}^{\rm chron}.
 \label{eq:chron-likelihood-definition}
\end{equation}
Its operator-Schmidt decomposition is
\begin{equation}
 L_{n,\ell}^{\rm chron}
 =
 \sum_{r=0}^{\ell}
 \sqrt{\frac{(\ell)_r}{(n)_r}}
 \sum_{|\bm m|=r}
 (F_{\bm m}^{(n)})^T
 \otimes
 F_{\bm m}^{(\ell)}.
 \label{eq:chron-schmidt-app}
\end{equation}
The transpose is taken in the computational basis used to define the Choi state. Thus the degree-$r$ singular value is
\[
\lambda_{n,\ell,r}
=
\sqrt{\frac{(\ell)_r}{(n)_r}},
\]
with multiplicity $\binom{r+2}{2}$.

Orthogonality of the symmetric Pauli sums then gives
\begin{equation}
 G_{n,\ell}^{\rm chron}
 =2^{-(n+\ell)}\Tr\left[(L_{n,\ell}^{\rm chron})^2\right]
 =
 \sum_{r=0}^{\ell}
 \binom{r+2}{2}
 \frac{(\ell)_r}{(n)_r}.
 \label{eq:chron-purity-app}
\end{equation}
Suppose that
\[
\frac{\ell}{n}\longrightarrow\alpha<1.
\]
Extend the summand in \eqref{eq:chron-purity-app} by zero for $r>\ell$, with the convention $(\ell)_r=0$ in that range. For each fixed $r\geq0$,
\[
\frac{(\ell)_r}{(n)_r}
\longrightarrow
\alpha^r.
\]
Choose $\alpha_+$ with
\[
\alpha<\alpha_+<1.
\]
For all sufficiently large $n$, $\ell/n\leq\alpha_+$, and
\begin{equation}
0
\leq
\frac{(\ell)_r}{(n)_r}
=
\prod_{j=0}^{r-1}
\frac{\ell-j}{n-j}
\leq
\left(\frac{\ell}{n}\right)^r
\leq
\alpha_+^r.
\end{equation}
Since
\[
\sum_{r=0}^{\infty}
\binom{r+2}{2}\alpha_+^r
<
\infty,
\]
dominated convergence yields
\begin{equation}
 G_{n,\ell}^{\rm chron}
 \longrightarrow
 \sum_{r=0}^{\infty}
 \binom{r+2}{2}\alpha^r
 =
 \frac{1}{(1-\alpha)^3}.
 \label{eq:chron-purity-limit-app}
\end{equation}

These identities give the exact singular spectrum and purity limit used in Sec.~\ref{sec:chronological}.

\end{document}